\documentclass[10pt, a4paper, logo]{googledeepmind} %
\usepackage{times}

\usepackage{hyperref}
\usepackage[leftcaption]{sidecap} 

\usepackage{amsmath,amsfonts,bm}
\usepackage{natbib}
\usepackage{fancyhdr}

\def\eqref#1{equation~\ref{#1}}

\def\1{\bm{1}}

\DeclareMathAlphabet{\mathsfit}{\encodingdefault}{\sfdefault}{m}{sl}
\SetMathAlphabet{\mathsfit}{bold}{\encodingdefault}{\sfdefault}{bx}{n}

\usepackage{hyperref}
\usepackage{url}
\usepackage{booktabs}
\usepackage{graphicx}
\usepackage{subcaption}
\usepackage{amssymb}
\usepackage[most]{tcolorbox}
\usepackage[table]{xcolor}
\tcbuselibrary{breakable} 
\usepackage{amsmath} 
\usepackage{siunitx}
\usepackage{colortbl}

\usepackage{microtype}
\usepackage{color}
\usepackage{wrapfig}
\usepackage{natbib}
\usepackage{colortbl}
\usepackage{microtype}
\usepackage{graphicx}
\usepackage{booktabs} 
\usepackage{array}
\usepackage{textcomp}
\usepackage{stfloats}
\usepackage{float}
\usepackage{verbatim}
\usepackage[ruled, linesnumbered, lined]{algorithm2e}
\usepackage{multirow}
\usepackage{enumitem}

\definecolor{BestColor}{HTML}{C8E6C9}  
\definecolor{SecondBestColor}{HTML}{FFF9C4} 

\usepackage{tcolorbox}
\usepackage{amsmath,amsfonts}

\definecolor{ggg}{RGB}{26,179,0}
\definecolor{rrr}{RGB}{179,0,0}
\definecolor{oodc}{RGB}{31,73,121}
\definecolor{idc}{RGB}{68,142,68}

\definecolor{mygray}{gray}{0.9}

\usepackage{makecell}
\newtheorem{theorem}{Theorem}
\usepackage{amsthm}
\usepackage{algorithmic}

\newtheorem{proposition}{Proposition} 

\def\Bias#1#2{\bm{b}}

\newtcolorbox{examplebox}[2][]{ 
    breakable, 
    enhanced, 
    colback=white, 
    colframe=cyan, 
    coltitle=white, 
    fonttitle=\bfseries, 
    title=#2, 
    overlay middle={\draw[cyan, line width=1pt](frame.south west)--(frame.south east);}, 
    overlay last={\draw[cyan, line width=1pt](frame.south west)--(frame.south east);}, 
    #1 
}

\usepackage[T1]{fontenc}
\usepackage{booktabs}      
\usepackage{graphicx}      
\usepackage[table]{xcolor} 
\usepackage{siunitx}       
\usepackage{etoolbox}      
\usepackage[normalem]{ulem}     

\definecolor{impcolor}{HTML}{2E8B57} 

\newcommand{\improvementstyle}[1]{$^{\textcolor{impcolor}{\tiny #1}}$}

\newcommand{\scoreimp}[2]{%
  \textbf{#1}%
  \ifstrequal{#2}{+0.0}{}{%
    \ifstrequal{#2}{0.0}{}{%
      \makebox[0pt][l]{\improvementstyle{#2}}%
    }%
  }%
}

\title{ECHO: Elastic Speculative Decoding with Sparse Gating for High-Concurrency Scenarios}

\author[1]{Xinyi Hu$^{*}$}
\author[1,2]{Yuhao Shen$^{*}$}
\author[1]{Baolin Zhang$^{*}$}
\author[1]{Hengxin Zhang}
\author[1]{Jun Dai}
\author[1]{Shuang Ge\textsuperscript{\dag}}
\author[1]{Lei Chen}
\author[1]{Yue Li}
\author[1]{Mingcheng Wan}
\affil[1]{Qwen Applications Business Group of Alibaba}
\affil[2]{Zhejiang University}

\begin{abstract}
Speculative Decoding promises to accelerate the inference of Large Language Models, yet its efficacy often degrades in production-grade serving. Existing evaluations typically overlook the compute-bound nature of high-concurrency regimes, where verification compute becomes the dominant bottleneck. Consequently, prior methods face a dilemma: static trees incur massive verification waste, while dynamic trees suffer from cumulative misjudgments and kernel incompatibility.
To bridge this gap, we introduce \textit{ECHO}, a high concurrency-oriented framework integrated into \texttt{SGLang} that reformulates speculative execution as a budgeted scheduling problem.
Crucially, \textit{ECHO} employs sparse confidence gating to manage the batch as a unified super-tree, elastically pivoting budget between depth and width to co-optimize the trade-off between reducing global verification steps and maximizing per-step efficiency.
Extensive evaluations across diverse model scales—particularly the industrial-grade Qwen3-235B—demonstrate that \textit{ECHO} consistently outperforms SOTA methods in both low-load and high-load scenarios, achieving up to 5.35$\times$ walltime speedup and delivering over 20\% relative speedup gain.
\end{abstract}

\begin{document}
\maketitle
\section{Introduction}
As Large Language Model (LLM) scales toward massive parameter counts, high-concurrency, and long-context scenarios, the inefficiency of autoregressive (AR) decoding becomes a first-order bottleneck for production-level system~\citep{yang2025qwen3,guo2025deepseek,singh2025openai,gemini2025modelcard,sadhukhan2024magicdec}. 
Speculative Decoding (SD) mitigates this serial barrier via a draft then verify paradigm~\citep{leviathan2023fast, chen2023accelerating}: a lightweight draft model predicts multiple tokens ahead, and the target model verifies them in parallel, amortizing expensive target computation across multiple accepted tokens. 
Tree-based drafting further enlarges the candidate set within each verification step~\citep{li2025eagle,miao2024specinfer}, often improving the accepted length per step in low-concurrency settings. 
However, we observe that these gains degrade precipitously in high-concurrency scenarios scenarios, as shown in Figure~\ref{fig:batch_scaling}.

\begin{figure}[h]
\begin{center}
\centerline{\includegraphics[width=0.8\columnwidth]{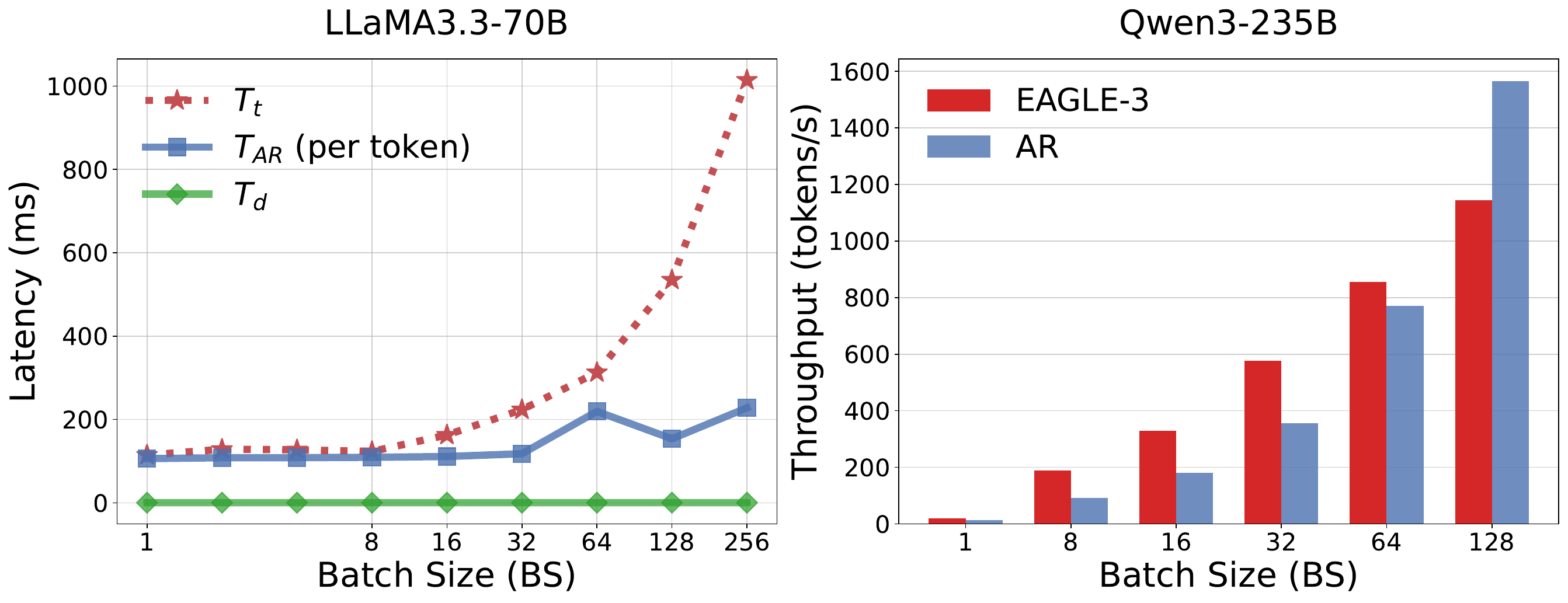}}
\caption{\textbf{Performance degradation in high-concurrency scenarios on MT-Bench.}
\textbf{Left:} Latency breakdown for LLaMA-3.3-70B. Verification cost ($T_t$) scales linearly, becoming the dominant bottleneck over AR cost ($T_{AR}$) as batch size increases.
\textbf{Right:} Throughput for Qwen3-235B. Due to this bottleneck, the speedup of EAGLE-3 diminishes and eventually underperforms vanilla AR at high concurrency (BS=128).
} 
\label{fig:batch_scaling}
\end{center} 
\vspace*{-0.3in}
\end{figure}

This degradation is not incidental: modern serving operates in an increasingly compute-bound regime, where many requests contend for the target model's verification compute, and any wasted verified token directly translates into lost goodput and worse tail latency.
This motivates a broader question beyond any single tree heuristic:

\textbf{\textit{How can we establish a principled methodology for SD in LLM serving, where verification compute is the primary bottleneck and system-level constraints dictate whether algorithmic efficiency translates into actual goodput gains?}}

Revisiting prior work reveals a core conflict in high-concurrency scenarios: the trade-off between the verification waste of static methods and the cost of misjudgment in dynamic methods.
Static tree methods~\citep{miao2024specinfer,li2024eagle2,li2025eagle} use fixed tree structures. While simple, they often verify too many useless branches when the model is uncertain. This leads to significant waste: although they aim to reduce the total number of steps, the large tree size makes each single verification step overly time-consuming.
Dynamic tree strategies~\citep{liu2026talon,brown2024dynamic} try to reduce this waste by adjusting the tree online, but they face severe challenges in production. First, prediction errors often accumulate in high-concurrency settings, canceling out the expected speed gains. Second, they create irregular batch shapes, which are not supported by modern SD frameworks.
System schedulers~\citep{wu2025tetris,liu2025turbospec} manage global budgets but typically treat the draft tree as a ``black box,'' ignoring the inefficiency within the tree construction itself.
In short, the community still lacks a production-ready solution that minimizes both verification waste and the extra costs of dynamic control.

\begin{figure}[h]
\begin{center}
\centerline{\includegraphics[width=0.8\columnwidth]{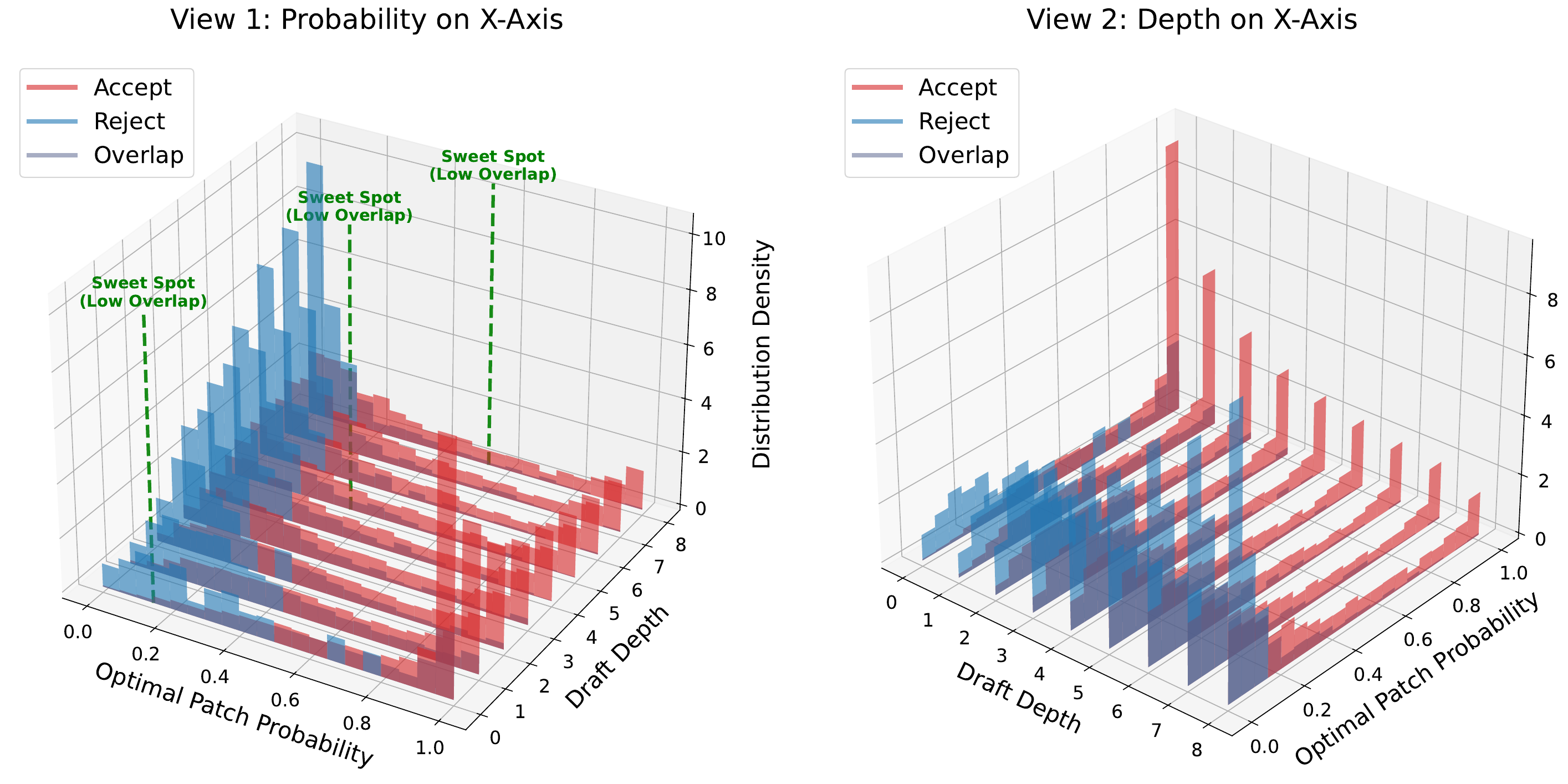}}
\caption{\textbf{Visualization of confidence distributions across draft depths} (LLaMA3.1-8B on MT-Bench). The two views show the probability of accepted and rejected draft tokens at different tree depths. Low overlap indicates better discriminability of the confidence scores, so these layers are identified as ``Sweet Spots''.}
\label{fig:Layer-wise confidence density}
\end{center} 
\vspace*{-0.3in}
\end{figure}

It becomes clear that effective SD in production should take care of two factors: (1) \textbf{reducing global verification steps}, and (2) \textbf{increasing per-step verification efficiency}. Prior methods typically optimize one metric at the expense of the other, lacking a precise mechanism to reconcile this delicate trade-off.
To resolve this tension, we first investigate why the per-step efficiency gains from dynamic adjustment fail to offset the overhead of increased verification steps. 
We identify the root cause as the frequent misjudgment present in current expansion policies.
As illustrated in Figure~\ref{fig:Layer-wise confidence density}, the discriminative power of confidence scores is inherently depth-dependent: the root and target depth consistently act as high-fidelity ``sweet spots'' with sharp signal-to-noise ratios, whereas most intermediate layers exhibit blurred decision boundaries. Consequently, prior work~\citep{huang2025adaspec,li2025adaserve,hong2025inference} compels decisions at noisy intermediate layers, introducing cumulative misjudgment risks that negate potential gains.

Driven by this insight, we present \textit{ECHO} (\textit{Elastic Speculative Decoding with Sparse Gating for High-Concurrency Scenarios}), a framework that reformulates tree construction as a budget scheduling problem. To align with the requirements of high-concurrency scenarios, \textit{ECHO} operates under a strict fixed verification budget. 
Within this fixed envelope, \textit{ECHO} employs a sparse gating strategy anchored on high-fidelity "sweet spots"--primarily the root and target depth, supplemented by a few discriminative intermediate layers selected dynamically. Specifically, we trigger a "Truncate-to-Widen" operation at these checkpoints to maximize yield under uncertainty, or an "Extend" operation to capitalize on momentum. 
\textbf{This selective elasticity enables \textit{ECHO} to aggressively deepen the candidate tree to reduce global verification steps when confidence is high, while confidence is insufficient, pivoting to width to safeguard per-step efficiency.}

Crucially, we implement \textit{ECHO} atop the industrial-grade \texttt{SGLang} framework. Unlike prior dynamic tree methods that primarily report results on raw \texttt{transformers}—leaving their incompatibility with standard serving kernels unaddressed—we explicitly implemented specialized operators to support irregular batch shapes within \texttt{SGLang}. This system-level adaptation allows the tree structure to fluidly morph to match the request's entropy, effectively bridging the gap between theoretical dynamic algorithms and production deployment.

Our main contributions are as follows:
\begin{itemize}[leftmargin=*, itemsep=0pt, topsep=4pt]
    \item We formulate a serving-centric methodology for SD in modern LLM serving, highlighting verification compute as the dominant scarce resource and identifying two primary failure modes undersystem constraints: verification waste and misjudgment accumulation in dynamic control.

    \item We propose \textit{ECHO}, a training-free framework that employs sparse gating at high-fidelity "sweet spots" (root, target depth, and adaptive intermediate layers). By scheduling width vs. depth under a fixed budget, \textit{ECHO} minimizes misjudgment accumulation. 

    \item We validate \textit{ECHO} on various models, especially the large-scale Qwen3-235B. The results demonstrate that \textit{ECHO} consistently outperforms state-of-the-art (SOTA) methods, particularly in speculation-hard regimes where prior art degrades.

    \item We provide the implementation of dynamic-tree SD for high-concurrency scenarios and integrated it into \texttt{SGLang} for the first time, bridging the gap between academic research and production deployment.
\end{itemize}

\section{Problem Formulation}
\label{sec:problem}
\paragraph{Standard SD.}
SD pairs a lightweight draft model with a target model.
At step $t$ with prefix $x_{1:t}$, the draft model proposes $K$ draft tokens, and the target verifies them in one parallel forward pass, accepting a prefix of length $L\in\{0,\dots,K\}$ (e.g., via rejection sampling~\citep{leviathan2023fast}).

A common \texttt{Speedup} proxy is 
\begin{equation}
\textit{Speedup} =
\frac{(\mathbb{E}[L]+1)\, T_{\text{ar}}}{T_{\text{draft}}(K) + T_{\text{verify}}(K)},
\label{eq:standard_speedup}
\end{equation}
where $\mathbb{E}[L]$ represents the mean accepted tokens (MAT), $T_{\text{ar}}$ is the cost of one autoregressive step, and $T_{\text{draft}}/T_{\text{verify}}$ correspond to the drafting and verification latency.
Prior work often assumes $T_{\text{verify}} \approx T_{\text{ar}}$ for moderate $K$, motivating the use of large draft token trees to maximize $\mathbb{E}[L]$ regardless of the verification overhead~\citep{li2025eagle}.

\paragraph{SD in High-Concurrency Scenarios.}
In production serving, a batch contains $B$ concurrent requests.
Let request $i$ propose $K_i$ tokens. The total number of tokens to verify in one step is $K_{\text{total}}=\sum_{i=1}^{B}K_i$.
Recent studies~\citep{Liu2025Illusion,sadhukhan2024magicdec} show that verification becomes \textbf{compute-bound} when $K_{\text{total}}$ is large. This often happens with large models and long contexts. In this case, the latency grows linearly with the number of tokens:
\begin{equation}
T_{\text{ver}}(K_{\text{total}})\approx T_{\text{ar}}(1 +\gamma \cdot \left[ K_{\text{total}} - K_{\text{max}} \right]^+),
\label{eq:ver_linear}
\end{equation}
where $K_{\max}$ is the computing limit of the hardware.
This means the ``free-lunch'' assumption fails: once the total load exceeds $K_{\text{max}}$, any increase in $K_i$ linearly penalizes the verification latency for all concurrent requests.

\paragraph{Budget-Constrained Objective.}
In compute-bound regimes, SD is best viewed as a verification-budget allocation problem.
Under a strict budget, the objective is no longer solely to maximize $\mathbb{E}[L]$ via aggressive tree expansion (aiming to reduce global verification steps), but to maximize per-step verification efficiency.
We quantify this by draft \texttt{Yield}:
\begin{equation}
\texttt{Yield} = \frac{\mathbb{E}[L]}{1 + \left[ K_{\text{total}} - K_{\text{max}} \right]^+},
\label{eq:yield_def}
\end{equation}
i.e., the fraction of verified tokens that contribute to the accepted prefix.
Rigid static trees may achieve high $\mathbb{E}[L]$ but incur large verification waste via low-utility branches, while dense dynamic control can introduce decision errors and overhead that accumulate under concurrency.
This motivates a serving-oriented design that improves \texttt{Yield} under a fixed verification budget with minimal control.

\begin{figure*}[t]
\begin{center}
\centerline{\includegraphics[width=1\columnwidth]{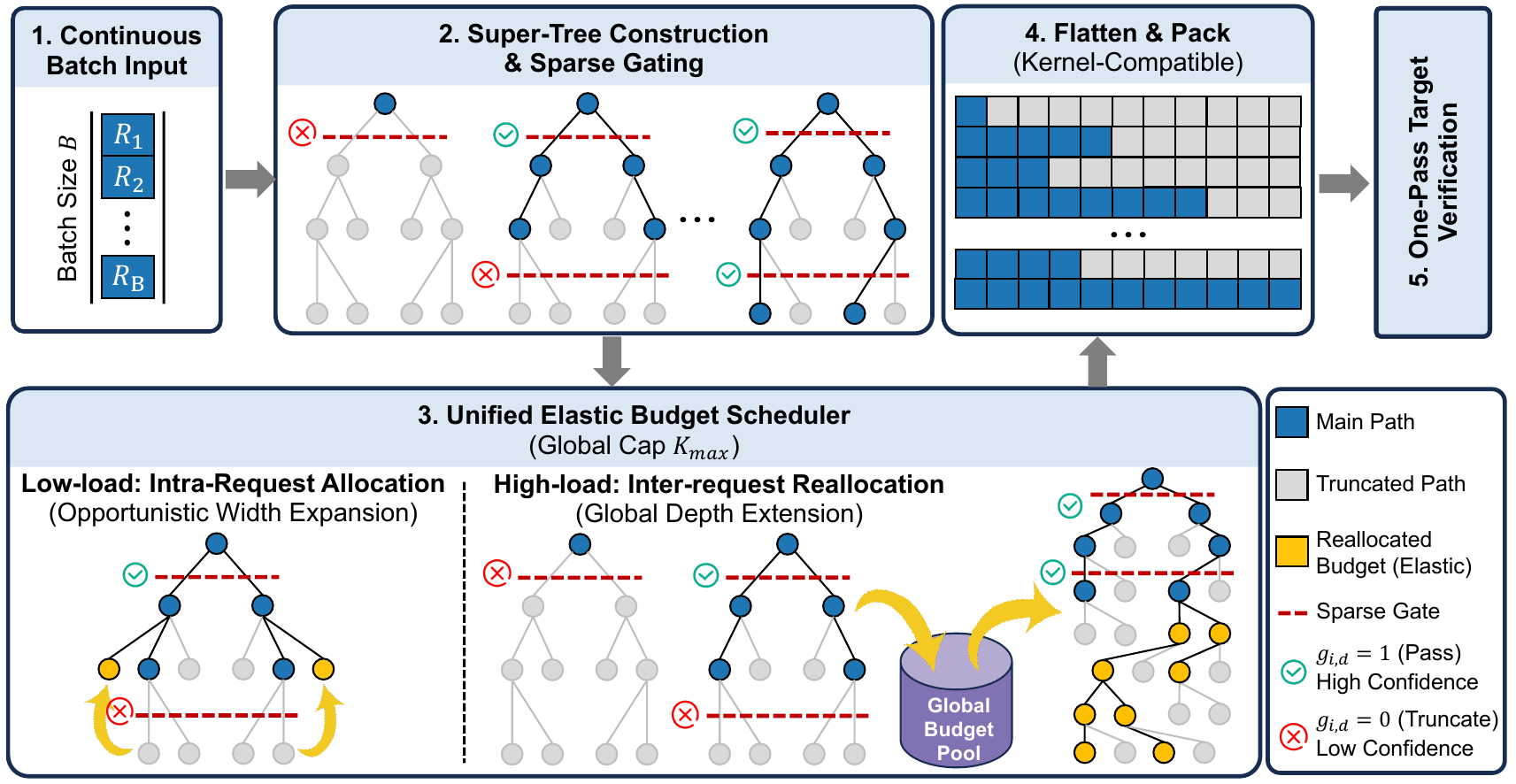}} 
\caption{\textbf{Overview of the \textit{ECHO} Framework.}
\textbf{(1) Super-Tree Construction:} Draft trees are evaluated (truncated or extended) only at sparse gates.
\textbf{(2) Unified Elastic Budget Scheduler:} Under a global verification cap ($K_{\max}$), the scheduler dynamically adapts resource allocation.
In Low-Load scenarios, the budget saved by truncation is reused locally to \textbf{widen} the current tree.
In High-Load scenarios, budget saved from truncated low-confidence requests is reallocated to extend the \textbf{depth} of high-confidence requests.
\textbf{(3) Flatten \& Pack:} Finally, the ragged batch formed by requests with varying token counts is packed into a dense, kernel-compatible layout for efficient verification.
}

\label{fig:overview}
\end{center}
\vspace*{-0.3in}
\end{figure*}

\section{Method}
\label{sec:method}

\subsection{Overview: The \textit{ECHO} Framework}
\label{subsec:method_overview}
We propose \textit{ECHO}, a budget-aware speculative decoding framework that unifies tree construction (depth vs.\ width) and batch-level scheduling (request-to-request allocation) under a single serving-centric objective: maximize useful progress under a fixed verification compute budget.
In modern high-concurrency scenarios, verification easily becomes compute-bound, leaving a narrow margin for speedup. 
Consequently, the system is strictly constrained by the total number of tokens verified in each step. \textit{ECHO} therefore models each verification iteration as a global token-budget allocation problem.
Figure~\ref{fig:overview} illustrates the unified pipeline, showing how the unified scheduler coordinates dynamic tree construction and elastic resource reallocation before packing the batch for efficient execution.

\paragraph{Super-Tree View.}
Consider a batch of $B$ concurrent requests.
At a given SD iteration, each request $i$ constructs a candidate token tree $\mathcal{G}_i$.
Let $\mathcal{V}_i$ denote the set of all candidate nodes in $\mathcal{G}_i$.
The target model verifies the flattened union $\bigcup_{i=1}^B \mathcal{V}_i$ in one parallel forward pass~\cite{li2024eagle}.

We impose a batch-level verification cap:
\begin{equation}
\sum_{i=1}^{B} K_i \;\le\; K_{\max},
\qquad K_i \triangleq |\mathcal{G}_i|,
\label{eq:global_budget}
\end{equation}
where $K_i$ is the number of candidate tokens submitted for request $i$, and $K_{\max}$ is set near the hardware saturation threshold (i.e., the compute-bound limit).
Under this view, the batch behaves as a single super-tree sharing one global budget $K_{\max}$: allocating additional candidates to one request necessarily reduces the candidates available to others.

\paragraph{Core Components.}
\textit{ECHO} couples two key components that jointly determine \emph{where} and \emph{how} to spend the budget:
\begin{itemize}[leftmargin=*, itemsep=0pt, topsep=4pt]
    \item \textbf{Sparse Confidence Gating:} A lightweight reliability signal evaluated only at a sparse set of discriminative ``sweet spots'' (e.g., root, target depth, and select intermediate depths). This avoids the overhead and misjudgment accumulation of per-depth/node control.
    \item \textbf{Elastic Budget Scheduling:} A unified scheduler that reallocates the shared budget across two axes: (1) \textbf{within a request} (adjusting depth vs.\ width), and (2) \textbf{across requests} (shifting budget from one request to another), all constrained by Eq.~\ref{eq:global_budget}.
\end{itemize}

\subsection{Sparse Gating at Sweet Spots}
\label{subsec:sparse_gating}
Prior dynamic tree methods~\cite{liu2026talon, brown2024dynamic} often perform dense, fine-grained control (e.g., per-depth/node) to adjust topology.
However, our analysis indicates that the discriminative power of confidence scores is depth-dependent (Figure~\ref{fig:Layer-wise confidence density}): only a few depths act as reliable ``sweet spots'', while most depths exhibit blurred decision boundaries.
\textit{ECHO} exploits this structure by restricting gating to a sparse set of checkpoint depths.

\vspace*{-0.05in}
\paragraph{Path Scoring.}
Consider the draft tree for request $i$.
For a node $j \in \mathcal{V}_{i,d}$ at depth $d$ corresponding to token $x_{i,d,j}$, we define its path score (cumulative log-score) recursively:
\begin{equation}
S_{i,d,j}
=
S_{i,d-1,\mathrm{pa}(j)} + \log q\!\left(x_{i,d,j}\mid h_{i,d-1,\mathrm{pa}(j)}\right),
\label{eq:path_score_method}
\end{equation}
where $\mathrm{pa}(j)$ denotes the parent node index, $h_{i,d-1,\mathrm{pa}(j)}$ is the hidden state used to generate the current token.
\vspace*{-0.05in}
\paragraph{Layer Confidence.}
We summarize the reliability of the entire candidate set $\mathcal{V}_{i,d}$ at depth $d$ for request $i$ using the maximum-likelihood path probability:
\begin{equation}
c_{i,d}
=
\exp\!\left( \max_{j\in \mathcal{V}_{i,d}} S_{i,d,j}\right).
\label{eq:layer_conf_method}
\end{equation}
Intuitively, $c_{i,d}$ measures whether there exists at least one highly probable path up to depth $d$.
\vspace*{-0.05in}
\paragraph{Offline Calibration.}
Similarly to the calibration phase in Post-Training Quantization~\cite{xiao2023smoothquant,lin2024awq}, \textit{ECHO} determines the gating parameters during the warm-up period with almost no additional overhead. 
We quantify the separability of confidence scores $c_{i,d}$ for accepted versus rejected tokens using the AUC metric~\cite{hanley1982meaning}.
We identify ``sweet spots'' $\mathcal{D}_{\mathrm{sig}}$ as depths exhibiting high discriminative power (i.e., $\text{AUC}_d > \delta$) and calibrate the thresholds $\tau_d$ at these checkpoints to maximize the separation between the two distributions.
\vspace*{-0.05in}
\paragraph{Gating Signal.}
During inference, \textit{ECHO} restricts gating decisions strictly to the pre-calibrated checkpoints $d \in \mathcal{D}_{\mathrm{sig}}$.
At these depths, the system emits a binary control signal based on the learned thresholds:
\begin{equation}
g_{i,d} = \mathbb{I}\left[c_{i,d} > \tau_d\right].
\label{eq:gating_signal}
\end{equation}
Here, $g_{i,d}=1$ signals high confidence (safe to extend depth), while $g_{i,d}=0$ signals low confidence (high probability of rejection), triggering the scheduler to truncate and reallocate budget.

\subsection{Unified Elastic Budget Scheduling}
\label{subsec:budget_scheduling}
The core of \textit{ECHO} treats the verification cap $K_{\max}$ as a shared global pool.
To ensure fair resource competition, the scheduler operates depth-by-depth across the batch, applying a unified priority rule to resolve the key allocation trade-off:
whether to reinvest locally by widening the current draft tree at the truncation depth, or to reallocate globally to support other high-potential requests by allocating a larger token budget.
\vspace*{-0.05in}
\paragraph{Priority-Based Allocation.}
At decision step, we allocate the remaining budget based on a strict priority hierarchy:
\begin{itemize}[leftmargin=*, itemsep=0pt, topsep=4pt]
    \item \textbf{Priority 1: Global Depth Extension.}
    This is the primary objective, \textbf{aimed at aggressively reducing the total verification steps}. As long as any request in the batch maintains high confidence ($g_{i,d}=1$), the global budget is prioritized for its depth extension. If a request is truncated ($g_{i,d}=0$), it yields the budget to other active requests ($j \neq i$) that are still deepening.
    
    \item \textbf{Priority 2: Opportunistic Width Expansion.}
    This is triggered only when no active requests can further extend depth (e.g., all requests are truncated). Only then is the surplus budget reallocated to widen the candidate set of the truncated requests to improve coverage.
\end{itemize}
In a resource-contended batch, the budget is reserved for the depth expansion of confident requests before any width expansion is permitted.

\vspace*{-0.1in}
\paragraph{Elastic Scheduling.}
The unified priority rule adapts naturally to different load conditions (Alg.~\ref{alg:echo_scheduler}):
\begin{itemize}[leftmargin=*, itemsep=0pt, topsep=4pt]
    \item \textbf{Case 1 (Low-load: Intra-request Allocation).}
    When the budget is ample (e.g., $B=1$ or low contention), Priority 1 cannot fully absorb the available capacity.
    The scheduler naturally falls back to Priority 2, reinvesting the surplus locally to widen the draft tree at the truncation depth.
    This ``Truncate-and-Widen'' behavior enhances robustness when depth extension is risky.

    \item \textbf{Case 2 (High-load: Inter-request Reallocation).}
    When the batch is saturated (high contention), Priority 1 dominates.
    A truncated request ($g_{i,d}=0$) yields its potential budget to the global pool, where it is immediately claimed for the depth extension of other high-confidence requests.
    This mechanism effectively reallocates compute from low-confidence to high-potential requests to maximize aggregate throughput.
\end{itemize}

\begin{algorithm}[t]
\caption{\textit{ECHO}: Unified Elastic Budget Scheduling}
\label{alg:echo_scheduler}
\begin{algorithmic}[1]
\STATE \textbf{Input:} Batch $B$, Budget $K_{\max}$, Gating params $(\mathcal{D}_{\mathrm{sig}}, \tau)$, Max Width $W_{\max}$, Default Width $W_\text{topk}$
\STATE \textbf{Output:} Draft Trees $\{\mathcal{G}_i\}_{i=1}^B$
\STATE \textbf{Init:} $budget \leftarrow K_{\max}, \; d \leftarrow 0, \; active\_set \leftarrow \{1..B\}, \; trunc\_set \leftarrow \emptyset$

\STATE \COMMENT{\textbf{Phase 1: Global Depth Extension}}
\WHILE{$budget > 0$ \textbf{and} $active\_set \neq \emptyset$}
    \STATE $d \leftarrow d + 1$
    \FOR{$i \in active\_set$ \textbf{if} $budget > 0$}
        \STATE $is\_pass \leftarrow (c_{i,d} > \tau_d)$ \textbf{if} $d \in \mathcal{D}_{\mathrm{sig}}$ \textbf{else} True
        \IF{$is\_pass$}
            \STATE $\mathcal{G}_i.append(topk\_tokens) $
            \STATE $budget \leftarrow budget - W_\text{topk}$
        \ELSE
            \STATE $active\_set.remove(i); \quad trunc\_set.add(i)$
        \ENDIF
    \ENDFOR
\ENDWHILE

\STATE \COMMENT{\textbf{Phase 2: Opportunistic Width Expansion}}
\FOR{$i \in trunc\_set$ \textbf{if} $budget \ge k$}
    \STATE $w \leftarrow \min(budget, W_{\max})$ 
    \STATE $\mathcal{G}_i.widen(k); \quad budget \leftarrow budget - W_{\max}$
\ENDFOR
\end{algorithmic}
\end{algorithm}

\vspace*{-0.05in}

\section{Theoretical Guarantees}
\label{sec:theory_summary}
We summarize two key guarantees that motivate \textit{ECHO}'s design; full derivations are deferred to Appendix~\ref{app:theory}.
\vspace*{-0.1in}
\paragraph{Width improves coverage at a truncated depth.}
When depth extension is halted at depth $d$ (low confidence), \textit{ECHO} uses remaining budget to widen the frontier.
The next theorem states that widening strictly increases the probability of covering the target token at that depth.

\begin{theorem}[Coverage Gain via Width]
\label{thm:width_coverage}
Let $\mathcal{S}_k=\{x^{(1)},\dots,x^{(k)}\}$ be the top-$k$ candidate tokens at depth $d$ ranked by the target next-token distribution $p_t(x\mid x_{<d})$.
Let $x^*$ be the ground-truth token sampled from $p_t$, and define the coverage probability as $\mathbb{P}(x^* \in \mathcal{S}) \triangleq \sum_{x \in \mathcal{S}} p_t(x)$.
If we expand from $k$ to $k'>k$, then the coverage probability increases by the added probability mass:
\begin{equation}
\mathbb{P}(x^*\in \mathcal{S}_{k'})-\mathbb{P}(x^*\in \mathcal{S}_{k})
=\sum_{i=k+1}^{k'} p_t(x^{(i)}\mid x_{<d}) \;>\; 0,
\label{eq:coverage_gain}
\end{equation}
whenever $\sum_{i=k+1}^{k'} p_t(x^{(i)}\mid x_{<d})>0$.
\end{theorem}

\paragraph{Compute-bound serving favors reallocating budget to higher marginal utility.}
Under saturated compute-bound serving, \textit{ECHO} enforces a per-iteration token cap (Eq.~\ref{eq:global_budget}) so that iteration time is dominated by verifying a fixed number of tokens.
Thus, improving end-to-end throughput reduces to increasing the batch-level expected accepted tokens per iteration minimizing total verification steps, not necessarily each request's mean acceptance length.
The next theorem formalizes a sufficient condition: moving budget from low marginal gain to high marginal gain strictly improves the batch objective.

\begin{theorem}[Marginal Utility Exchange]
\label{thm:exchange}
Consider a batch iteration with fixed verification budget $\sum_{i=1}^B K_i = K_{\max}$.
Let $f_i(k)\triangleq \mathbb{E}[L_i \mid K_i=k]$ be request $i$'s expected accepted tokens given $k$ verified candidates, and define the marginal gain $\Delta_i(k)\triangleq f_i(k)-f_i(k-1)$.
For any two requests $i,j$, if
\begin{equation}
\Delta_j(K_j+1) \;>\; \Delta_i(K_i),
\label{eq:exchange_cond}
\end{equation}
then reallocating one token from $i$ to $j$ (i.e., $K_i\!\leftarrow\!K_i-1,\,K_j\!\leftarrow\!K_j+1$) strictly increases the batch-level expected accepted tokens $\sum_{m=1}^B \mathbb{E}[L_m]$, and therefore improves end-to-end throughput in compute-bound serving.
\end{theorem}

\section{Experiments}
\label{sec:experiments}

\begin{table*}[t]
\centering
\caption{\textbf{Main results on five benchmarks under the low-load setting (BS = 1)}. Performance comparison between \textit{ECHO} and existing baselines across diverse model configurations. Bold numbers denote the best speedup.}
\resizebox{\linewidth}{!}{%
\begin{tabular}{c l cc cc cc cc cc c}
\toprule
\multirow{2}{*}{Models} & \multirow{2}{*}{Methods} & \multicolumn{2}{c}{HumanEval} & \multicolumn{2}{c}{GSM8K} & \multicolumn{2}{c}{CNN/DM} & \multicolumn{2}{c}{Alpaca} & \multicolumn{2}{c}{MT-Bench}&\multirow{2}{*}{Avg.}\\
\cmidrule(lr){3-4} \cmidrule(lr){5-6} \cmidrule(lr){7-8} \cmidrule(lr){9-10} \cmidrule(lr){11-12} & &MAT & Speedup & MAT & Speedup & MAT & Speedup & MAT & Speedup & MAT & Speedup\\
\midrule
\multirow{9}{*}{\makecell[c]{Vicuna-13B}}
& Lookahead & 1.73 & 1.69$\times$ & 1.88 & 1.79$\times$ & 1.48 & 1.44$\times$ & 1.49 & 1.44$\times$ & 1.67 & 1.63$\times$ & 1.60$\times$\\
& Sps       & 2.55 & 1.81$\times$ & 1.99 & 1.75$\times$ & 2.31 & 1.71$\times$ & 2.01 & 1.74$\times$ & 2.25 & 1.81$\times$ & 1.76$\times$\\
& Medusa & 2.78 & 2.25$\times$ & 2.63 & 2.12$\times$ & 2.09 & 1.65$\times$ & 2.44 & 1.96$\times$ & 2.58 & 2.08$\times$ & 2.01$\times$\\
& Hydra & 3.87 & 2.75$\times$ & 3.66 & 2.60$\times$ & 2.82 & 1.95$\times$ & 3.51 & 2.48$\times$ & 3.64 & 2.53$\times$ & 2.46$\times$\\
& OPT-Tree & 8.21 & 3.85$\times$ & 6.77 & 3.25$\times$ & 6.91 & 2.75$\times$ & 6.56 & 3.20$\times$ & 6.95 & 3.30$\times$ & 3.27$\times$\\
& DDD  & 8.95 & 4.95$\times$ & 6.21 & 3.92$\times$ & 6.02 & 3.45$\times$ & 5.98 & 3.59$\times$ & 6.05 & 3.95$\times$ & 3.97$\times$\\
& EAGLE3 & 8.49 & 4.81$\times$ & 6.82 & 3.85$\times$ & 6.41 & 3.38$\times$ & 6.49 & 3.65$\times$ & 6.83 & 3.89$\times$ & 3.92$\times$\\

&\cellcolor{gray!18}\textbf{\textit{ECHO}} & \cellcolor{gray!18}9.35 & \cellcolor{gray!18}\textbf{5.25}$\times$ & \cellcolor{gray!18}6.53 & \cellcolor{gray!18}\textbf{4.08}$\times$ & \cellcolor{gray!18}6.34 & \cellcolor{gray!18}\textbf{3.53}$\times$ & \cellcolor{gray!18}6.24 & \cellcolor{gray!18}\textbf{3.74}$\times$ & \cellcolor{gray!18}6.33 & \cellcolor{gray!18}\textbf{4.12}$\times$ & \cellcolor{gray!18}\textbf{4.14}$\times$ \\

\midrule
\multirow{4}{*}{\makecell[c]{LLaMA3.1-8B}}
& DDD & 6.95 & 4.18$\times$ & 6.02 & 4.07$\times$ & 5.08 & 3.10$\times$ & 6.53 & 4.08$\times$ & 5.98 & 3.70$\times$ & 3.83$\times$\\
& EAGLE3 & 7.18 & 4.02$\times$ & 6.50 & 3.94$\times$ & 5.46 & 3.02$\times$ & 7.06 & 4.01$\times$ & 6.43 & 3.63$\times$ & 3.72$\times$ \\
& \cellcolor{gray!18}\textbf{\textit{ECHO}} & \cellcolor{gray!18}7.22 & \cellcolor{gray!18}\textbf{4.30}$\times$ & \cellcolor{gray!18}6.28 & \cellcolor{gray!18}\textbf{4.10}$\times$ & \cellcolor{gray!18}5.24 & \cellcolor{gray!18}\textbf{3.26}$\times$ & \cellcolor{gray!18}6.81 & \cellcolor{gray!18}\textbf{4.19}$\times$ & \cellcolor{gray!18}6.23 & \cellcolor{gray!18}\textbf{3.86}$\times$ & \cellcolor{gray!18}\textbf{3.94}$\times$ \\

\midrule
\multirow{4}{*}{\makecell[c]{LLaMA3.3-70B}}
& DDD & 6.82 & 5.13$\times$ & 6.15 & 4.63$\times$ & 4.92 & 3.52$\times$ & 6.68 & 4.84$\times$ & 5.70 & 4.02$\times$ & 4.43$\times$\\
& EAGLE3 & 7.12 & 4.98$\times$ & 6.53 & 4.72$\times$ & 5.19 & 3.60$\times$ & 6.95 & 4.75$\times$ & 5.92 & 4.09$\times$ & 4.43$\times$\\

&\cellcolor{gray!18}\textbf{\textit{ECHO}} & \cellcolor{gray!18}7.07 & \cellcolor{gray!18}\textbf{5.35}$\times$ & \cellcolor{gray!18}6.41 & \cellcolor{gray!18}\textbf{5.08}$\times$ & \cellcolor{gray!18}5.10 & \cellcolor{gray!18}\textbf{3.79}$\times$ & \cellcolor{gray!18}6.84 & \cellcolor{gray!18}\textbf{4.94}$\times$ & \cellcolor{gray!18}5.94 & \cellcolor{gray!18}\textbf{4.32}$\times$ & \cellcolor{gray!18}\textbf{4.70}$\times$ \\

\midrule
\multirow{4}{*}{\makecell[c]{Qwen3-8B}}
& DDD & 3.65 & 2.54$\times$ & 3.72 & 2.32$\times$ & 3.08 & 2.13$\times$ & 3.22 & 2.06$\times$ & 3.48 & 2.35$\times$ & 2.28$\times$\\
& EAGLE3 & 3.91 & 2.37$\times$ & 3.94 & 2.35$\times$ & 3.28 & 1.98$\times$ & 3.46 & 2.09$\times$ & 3.71 & 2.20$\times$ & 2.20$\times$\\

&\cellcolor{gray!18}\textbf{\textit{ECHO}} & \cellcolor{gray!18}3.82 & \cellcolor{gray!18}\textbf{2.74}$\times$ & \cellcolor{gray!18}3.88 & \cellcolor{gray!18}\textbf{2.68}$\times$ & \cellcolor{gray!18}3.20 & \cellcolor{gray!18}\textbf{2.37}$\times$ & \cellcolor{gray!18}3.36 & \cellcolor{gray!18}\textbf{2.51}$\times$ & \cellcolor{gray!18}3.63 & \cellcolor{gray!18}\textbf{2.57}$\times$ & \cellcolor{gray!18}\textbf{2.57}$\times$ \\

\midrule
\multirow{4}{*}{\makecell[c]{Qwen3-32B}}
& DDD & 2.82 & 2.18$\times$ & 3.15 & 2.38$\times$ & 2.38 & 1.71$\times$ & 2.68 & 2.11$\times$ & 2.88 & 1.97$\times$ & 2.07$\times$\\
& EAGLE3 & 2.99 & 2.02$\times$ & 3.32 & 2.41$\times$ & 2.55 & 1.67$\times$ & 2.83 & 1.97$\times$ & 3.04 & 1.93$\times$ & 2.00$\times$\\

&\cellcolor{gray!18}\textbf{\textit{ECHO}} & \cellcolor{gray!18}2.95 & \cellcolor{gray!18}\textbf{2.37}$\times$ & \cellcolor{gray!18}3.29 & \cellcolor{gray!18}\textbf{2.72}$\times$ & \cellcolor{gray!18}2.48 & \cellcolor{gray!18}\textbf{1.93}$\times$ & \cellcolor{gray!18}2.74 & \cellcolor{gray!18}\textbf{2.31}$\times$ & \cellcolor{gray!18}2.95 & \cellcolor{gray!18}\textbf{2.25}$\times$ & \cellcolor{gray!18}\textbf{2.32}$\times$ \\

\midrule
\multirow{3}{*}{\makecell[c]{Qwen3-235B}}
& DDD & 2.20 & 1.88$\times$ & 2.48 & 1.68$\times$ & 2.05 & 1.49$\times$ & 2.18 & 1.78$\times$ & 2.28 & 1.99$\times$ & 1.77$\times$\\
& EAGLE3 & 2.41 & 1.82$\times$ & 2.73 & 1.71$\times$ & 2.02 & 1.35$\times$ & 2.35 & 1.72$\times$ & 2.54 & 1.83$\times$ & 1.69$\times$\\

&\cellcolor{gray!18}\textbf{\textit{ECHO}} & \cellcolor{gray!18}2.32 & \cellcolor{gray!18}\textbf{2.23}$\times$ & \cellcolor{gray!18}2.59 & \cellcolor{gray!18}\textbf{1.92}$\times$ & \cellcolor{gray!18}2.08 & \cellcolor{gray!18}\textbf{1.63}$\times$ & \cellcolor{gray!18}2.24 & \cellcolor{gray!18}\textbf{2.07}$\times$ & \cellcolor{gray!18}2.37 & \cellcolor{gray!18}\textbf{2.23}$\times$ & \cellcolor{gray!18}\textbf{2.02}$\times$ \\

\bottomrule
\end{tabular}
}
\label{tab:main_resultes}
\end{table*}

\subsection{Experimental Setting}
\paragraph{Datasets and Models.}
We evaluated \textit{ECHO} in various LLMs, including Vicuna-13B~\citep{vicuna2023}, LLaMA-3.1-8B, LLaMA-3.3-70B~\citep{grattafiori2024llama3herdmodels}, and the Qwen3 series (8B/32B/235B)~\citep{yang2025qwen3}. 
Following protocols established by EAGLE~\citep{li2024eagle} and Spec-Bench~\citep{spec-bench-xia-2024}, our evaluation spans five comprehensive benchmarks covering code generation, mathematical reasoning, summarization, and chat: HumanEval~\citep{humaneval-chen-2021}, GSM8K~\citep{Cobbe_Kosaraju_Bavarian_Hilton_Nakano_Hesse_Schulman_2021}, CNN/DM~\citep{Nallapati_Zhou_dos}, Alpaca~\citep{taori2023alpaca}, and MT-Bench~\citep{mt-bench-zheng-2023}.
\paragraph{Baselines and Implementation.}
We benchmark \textit{ECHO} against representative methods across four distinct categories of SD: 
(1) Standard SD~\citep{chen2023accelerating}; (2) Retrieval-based~\citep{fu2024break}; (3) Training-based~\citep{cai2024medusa,anknerhydra,li2025eagle}; and (4) Dynamic Trees~\citep{wang2025opt,brown2024dynamic}, which we adapted to support EAGLE-3 for fair comparison.
Experiments are conducted on 8$\times$H100 GPUs with greedy sampling. We use HuggingFace \texttt{transformers} for low-load (BS=1) and \texttt{SGLang} for high-load scenarios.
All baselines are reproduced using their official configurations, with detailed hyperparameters provided in the Appendix~\ref{sec:appendix_eval}.
\paragraph{Metrics}
Since \textit{ECHO} strictly adheres to the speculative sampling acceptance conditions and does not modify the target model's weights, the output distribution remains mathematically identical to the target model. Consequently, we focus exclusively on acceleration metrics rather than generation quality:
\begin{itemize}[leftmargin=*, itemsep=0pt, topsep=2pt]
    \item \textbf{Wall-time Speedup:} The actual end-to-end speedup ratio relative to vanilla autoregressive decoding.
    \item \textbf{Mean Accepted Tokens (MAT):} The average number of tokens accepted per verification cycle. However, we argue that MAT is insufficient for evaluating dynamic drafting methods, as arbitrarily increasing the draft depth inevitably inflates the MAT without necessarily reflecting computational efficiency~\citep{brown2024dynamic}.
    \item \textbf{Draft Utilization ($u$):} To address the limitations of MAT, we introduce Draft Utilization, defined as the ratio of accepted tokens to the total number of drafted tokens per cycle ($u = \text{MAT} / \text{Depth}$). This metric provides a more precise measure of the drafting strategy's efficiency.
\end{itemize}

\subsection{Main Results}
\label{subsec:main_results_bs1}

\paragraph{Low-Load Case ($\text{BS}=1$).}
\label{subsec:main_results_bs1}
Table~\ref{tab:main_resultes} summarizes the performance of \textit{ECHO} in the low-load regime. \textit{ECHO} establishes a new SOTA speedup range of \textbf{$1.63 \times$}--\textbf{$5.35 \times$} across all benchmarks, demonstrating three distinct advantages:
\noindent\textbf{(1) Scalability to Industrial Models.}
\textit{ECHO} excels on massive architectures where speculation is challenging. On the industrial-grade \textbf{Qwen3-235B}, it attains a \textbf{2.02$\times$} average speedup, surpassing dynamic (DDD, $1.77\times$) and static (EAGLE-3, $1.69\times$) baselines by \textbf{14\%} and \textbf{19\%}, respectively. This validates the robustness of sparse gating against the sharp probability distributions characteristic of large-scale models.
\noindent\textbf{(2) Efficiency via Sparse Gating.}
Unlike prior dynamic methods burdened by the overhead of dense, layer-wise evaluation, \textit{ECHO} employs lightweight sparse gating to mitigate misjudgment accumulation. This efficiency allows \textit{ECHO} to outperform the representative dynamic method (DDD) by \textbf{15.8\%} on Qwen3-32B, effectively redirecting the saved computational budget toward effective token generation.
\noindent\textbf{(3) High Draft Utilization and Robustness.}
As illustrated in Figure~\ref{fig:draft_u}, \textit{ECHO} maximizes Draft Utilization ($u$) by dynamically aligning draft depth with confidence. By proactively truncating uncertain branches, \textit{ECHO} avoids the computational waste observed in static or overly aggressive drafting. Consequently, it achieves not only a higher average $u$ but also a significantly narrower interquartile range compared to baselines. This tighter distribution confirms \textit{ECHO}'s superior robustness, ensuring consistent drafting precision across diverse inputs.

\begin{figure}[t]
\begin{center}
\centerline{\includegraphics[width=0.7\columnwidth]{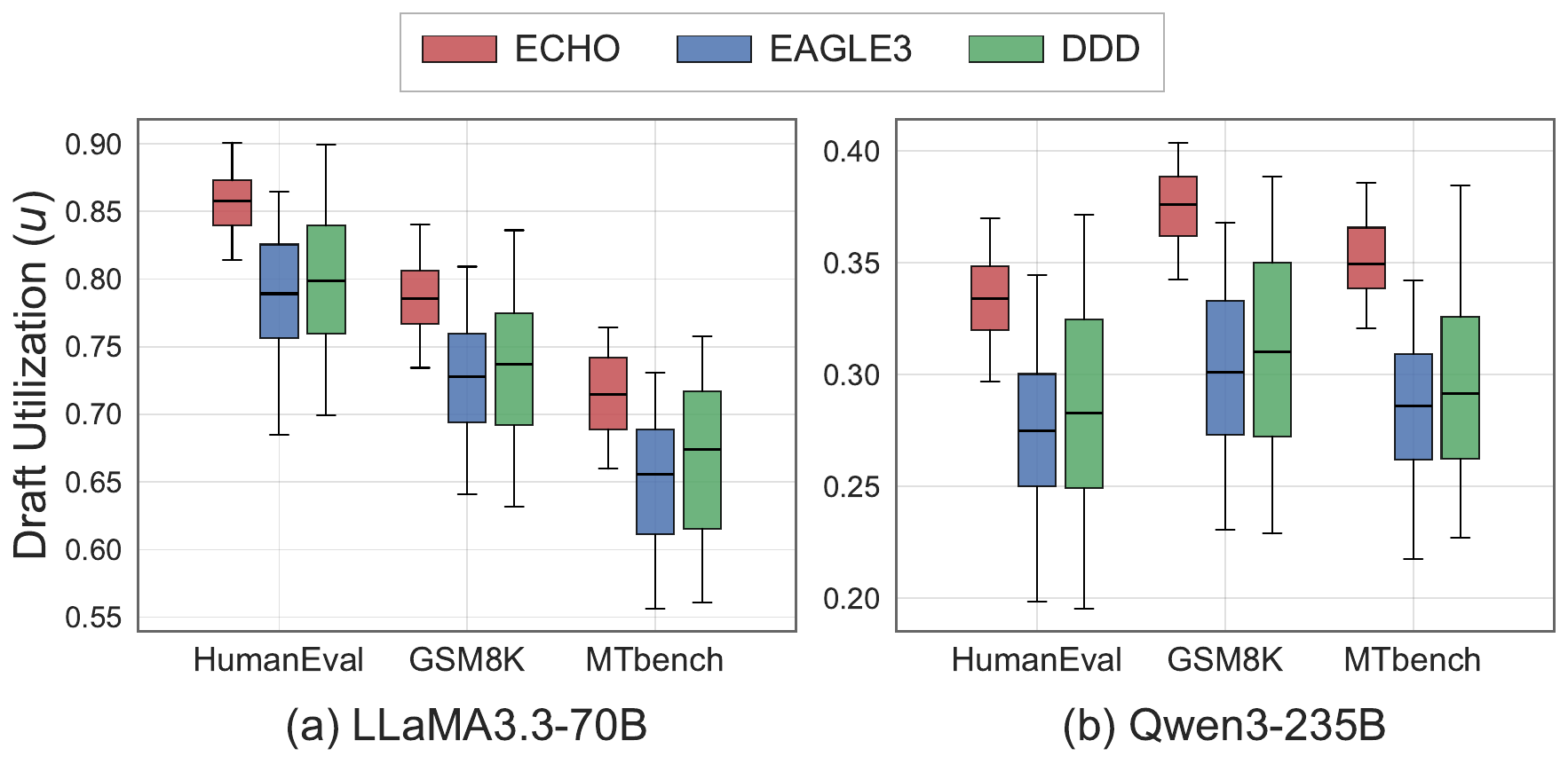}}
\caption{\textbf{Request-level Draft Utilization ($u$) across datasets and models.} We evaluate performance on LLaMA3.3-70B (a) and Qwen3-235B (b) across HumanEval, GSM8K, and MT-Bench. Each box shows the distribution of per-request Draft Utilization, spanning the 25th--75th percentiles, with whiskers covering the 5th--95th percentiles. \textit{ECHO} consistently surpasses EAGLE-3 and DDD in average utilization.}
\label{fig:draft_u}
\end{center} 
\vspace*{-0.4in}
\end{figure}

\paragraph{High-Load Case ($\text{BS}>1$).}
\label{subsec:main_results_bs}
We also evaluate the high-concurrency regime using \texttt{SGLang} to measure system throughput (tokens/s) across batch sizes from 8 to 256. Crucially, the efficacy of \textit{ECHO} is governed by \emph{when} the serving system enters the compute-bound verification regime. Smaller models remain memory-bound until high concurrency, whereas industrial-scale models saturate verification compute much earlier. Once saturation occurs, each iteration is effectively constrained by a fixed verification budget (Sec.~\ref{sec:problem}), and any inefficiency in (i) which tokens enter the verification batch or (ii) how budget is allocated across heterogeneous requests directly translates into lost throughput. \textit{ECHO} is a unified, budget-aware framework; however, different components dominate in different regimes. We analyze these two cases below.

\noindent\textbf{(1) Small-Scale Models: Late Compute-Bound and Precise Truncation.}
For smaller models (e.g., LLaMA-3.1-8B and Qwen3-8B), the system is often not compute-bound until high concurrency.
At low batch sizes, verifying extra speculative tokens hurts less because memory effects still dominate. Once verification becomes the bottleneck, EAGLE-3's fixed deep trees start wasting batch capacity on low-confidence branches.
\textit{ECHO} uses \textbf{sparse gating} to truncate these branches early, so fewer low-utility tokens enter verification. At $\text{BS}=256$, this improves throughput by \textbf{8\%} (10,703 $\rightarrow$ 11,551 tokens/s).

\noindent\textbf{(2) Industrial-Scale Models: Early Compute-Bound and Elastic Reallocation.}
For Qwen3-235B, verification becomes compute-bound at much lower concurrency, so wasted verification work immediately reduces global throughput. Here, the main issue is not only bad branches within one request, but also \emph{how budget is shared across requests}: EAGLE-3 spends too much on difficult requests while leaving little room to extend easy, high-confidence ones.
\textit{ECHO} applies \textbf{elastic budget scheduling} under the global cap (Eq.~\ref{eq:global_budget}), shifting tokens saved from truncated low-confidence requests to deepen high-confidence requests. This yields a \textbf{14.4\%} throughput gain at $\text{BS}=256$ (2,803 $\rightarrow$ 3,207 tokens/s), showing that budget reallocation is key for compute-bound serving.

\begin{figure*}[t]
\begin{center}
\centerline{\includegraphics[width=1\columnwidth]{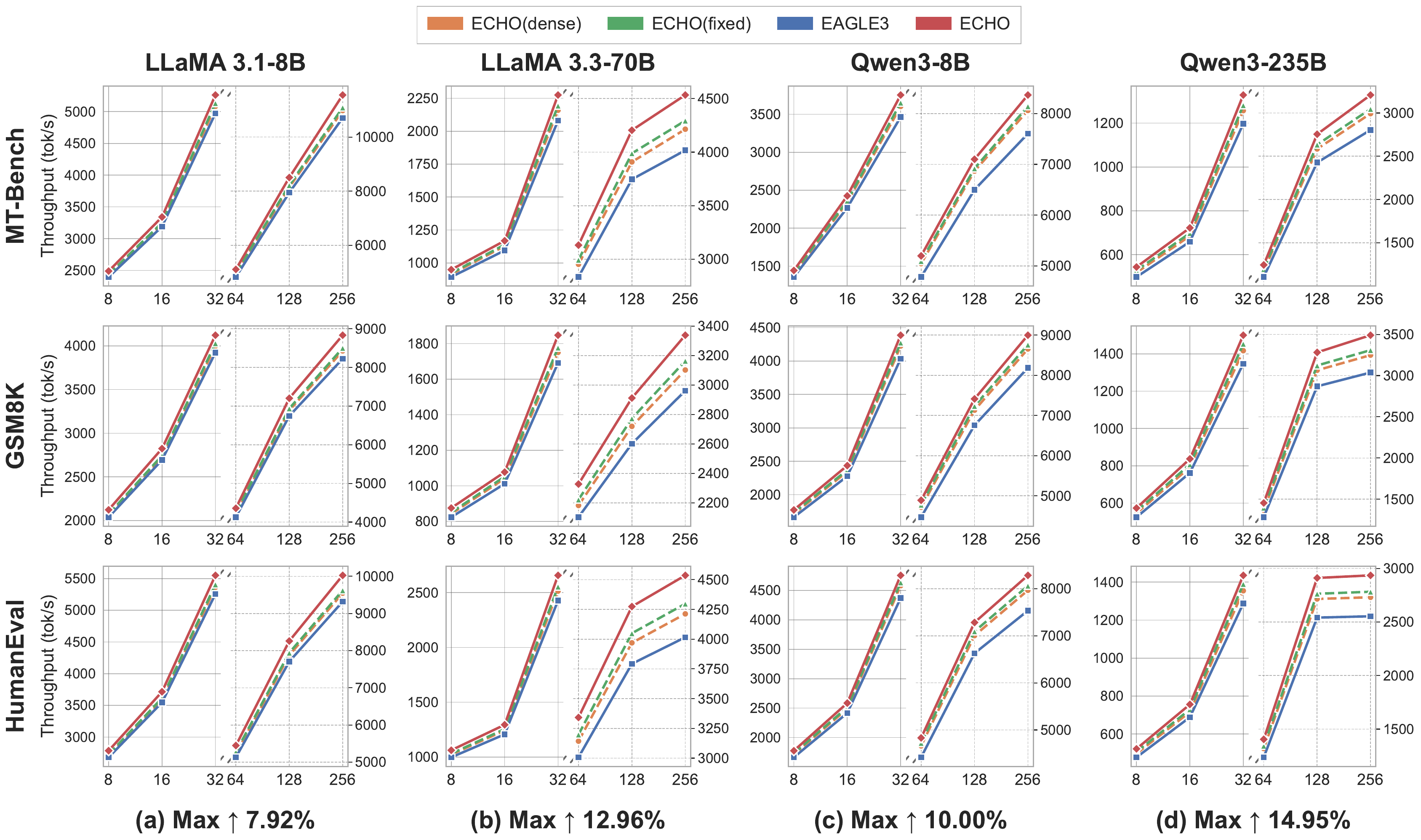}} 
\caption{\textbf{Main results on High-Load Case ($\text{BS} > 1$).} We evaluate \textit{ECHO} against EAGLE3 and two \textit{ECHO} variants on three benchmarks using four model configurations. The maximum improvement percentage below each column is against EAGLE3.}
\label{fig:results_bs>1}
\end{center}
\vspace*{-0.4in}
\end{figure*}

\subsection{Ablation}
To clarify the role of each component in \textit{ECHO}, we study two design questions: \emph{where} to apply gating (sparse vs. dense) and \emph{how} to set gating thresholds (depth-aware vs. fixed).
To understand which design choices matter, we compare \textit{ECHO} with two simplified variants: \textbf{Dense Gating} (making a gating decision for each depth) and \textbf{Fixed Threshold} (using a fixed threshold for all depths). Figure~\ref{fig:results_bs>1} shows that the full \textit{ECHO} consistently performs best.
\paragraph{Sparse Gating vs. Dense Gating.}
Dense gating is a natural baseline: if we check more often, can we save more wasted tokens?
The answer is no. On LLaMA-3.1-8B at $\text{BS}=256$, \textbf{Dense Gating} is about \textbf{5\%} worse than \textit{ECHO} (11,551 $\rightarrow$ 10,978 tokens/s).
This happens for two simple reasons. First, \emph{the checks themselves cost time}; doing them at every step adds up, especially for small models.
Second, confidence scores at many intermediate depths are \emph{not very reliable} (Figure~\ref{fig:Layer-wise confidence density}), so frequent checks can make wrong calls and cut branches that would have been accepted.
By only checking at a few reliable ``sweet spots'', \textit{ECHO} keeps the overhead low while making more accurate decisions.

\paragraph{Depth-Aware Threshold vs. Fixed Threshold.}
The \textbf{Fixed Threshold} baseline uses a single confidence cutoff $\tau$ for all depths.
This is brittle because confidence naturally decreases with depth: deeper tokens often have lower probabilities even when they are correct.
As a result, one threshold cannot work well everywhere: a high $\tau$ prunes too aggressively at deeper layers, while a low $\tau$ lets in too many low-quality tokens near the root.
On Qwen3-235B, \textit{ECHO} improves throughput over \textbf{Fixed Threshold} by \textbf{5.3\%} (3,046 $\rightarrow$ 3,207 tokens/s), showing that depth-aware calibration of $\tau_d$ is important for using the verification budget effectively.

\section{Related Works}
\label{sec:related_work}
\paragraph{Speculative Decoding}
SD leverages parallel verification to accelerate inference~\citep{leviathan2023fast, chen2023accelerating}. Existing approaches fall into three categories: draft model-based methods using separate lightweight drafters~\citep{zhou2023distillspec}; draft model-free methods utilizing auxiliary heads (e.g., Medusa~\citep{cai2024medusa}, EAGLE~\citep{li2025eagle}) or multi-token prediction~\citep{zeng2025glm}; and non-parametric methods relying on retrieval or matching logic~\citep{rest-he-2024}. While these works optimize draft generation, maximizing efficiency further requires optimizing the structure of these candidates.
\paragraph{Dynamic Token Tree}
Tree-based SD explores multiple paths to maximize acceptance~\citep{miao2024specinfer}. Static tree methods employ fixed structures for low overhead but lack flexibility~\citep{li2025eagle}. Conversely, dynamic tree strategies (e.g., TALON~\citep{liu2026talon}, OPT-Tree~\citep{wang2025opt}) adapt topology based on token probability. However, these methods typically rely on dense evaluation logic, introducing significant control overhead and generating irregular (ragged) batch shapes incompatible with standard high-performance serving kernels.
\paragraph{SD in High-Concurrency Scenarios}
Research has recently expanded from single-batch acceleration to system-level scheduling. While server-side~\citep{fu2024break} and client-side~\citep{liu2024pearl} optimizers exist, recent works like TETRIS~\citep{wu2025tetris} and TurboSpec~\citep{liu2025turbospec} have begun integrating SD into serving. Yet, prior art often treats the draft tree as a black box, overlooking the compute-bound constraints of high-concurrency regimes. \textit{ECHO} addresses this by jointly optimizing tree construction and budget scheduling under strict verification limits.

\section{Conclusion}
\label{sec:conclusion}
We study SD under high-concurrency scenarios and show that verification becomes the bottleneck, invalidating ``free-lunch'' assumptions. We propose \textit{ECHO}, a budget-aware framework that uses sparse gating and elastic scheduling to reduce verification waste while avoiding the error accumulation of fine-grained dynamic control. Experiments show consistent throughput gains across model scales, motivating a serving-oriented SD design that jointly optimizes global verification steps and per-step verification efficiency under compute-bound constraints.

\bibliography{iclr2026_conference}

@inproceedings{miao2024specinfer,
  title={Specinfer: Accelerating large language model serving with tree-based speculative inference and verification},
  author={Miao, Xupeng and Oliaro, Gabriele and Zhang, Zhihao and Cheng, Xinhao and Wang, Zeyu and Zhang, Zhengxin and Wong, Rae Ying Yee and Zhu, Alan and Yang, Lijie and Shi, Xiaoxiang and others},
  booktitle={Proceedings of the 29th ACM International Conference on Architectural Support for Programming Languages and Operating Systems, Volume 3},
  pages={932--949},
  year={2024}
}

@article{guo2025deepseek,
  title={Deepseek-r1: Incentivizing reasoning capability in llms via reinforcement learning},
  author={Guo, Daya and Yang, Dejian and Zhang, Haowei and Song, Junxiao and Zhang, Ruoyu and Xu, Runxin and Zhu, Qihao and Ma, Shirong and Wang, Peiyi and Bi, Xiao and others},
  journal={arXiv preprint arXiv:2501.12948},
  year={2025}
}

@article{cai2024medusa,
  title={Medusa: Simple llm inference acceleration framework with multiple decoding heads},
  author={Cai, Tianle and Li, Yuhong and Geng, Zhengyang and Peng, Hongwu and Lee, Jason D and Chen, Deming and Dao, Tri},
  journal={arXiv preprint arXiv:2401.10774},
  year={2024}
}

@article{li2024eagle,
  title={Eagle: Speculative sampling requires rethinking feature uncertainty},
  author={Li, Yuhui and Wei, Fangyun and Zhang, Chao and Zhang, Hongyang},
  journal={arXiv preprint arXiv:2401.15077},
  year={2024}
}

@article{li2024eagle2,
  title={Eagle-2: Faster inference of language models with dynamic draft trees},
  author={Li, Yuhui and Wei, Fangyun and Zhang, Chao and Zhang, Hongyang},
  journal={arXiv preprint arXiv:2406.16858},
  year={2024}
}

@article{fu2024break,
  title={Break the sequential dependency of llm inference using lookahead decoding},
  author={Fu, Yichao and Bailis, Peter and Stoica, Ion and Zhang, Hao},
  journal={arXiv preprint arXiv:2402.02057},
  year={2024}
}

@article{chen2023accelerating,
  title={Accelerating large language model decoding with speculative sampling},
  author={Chen, Charlie and Borgeaud, Sebastian and Irving, Geoffrey and Lespiau, Jean-Baptiste and Sifre, Laurent and Jumper, John},
  journal={arXiv preprint arXiv:2302.01318},
  year={2023}
}

@inproceedings{leviathan2023fast,
  title={Fast inference from transformers via speculative decoding},
  author={Leviathan, Yaniv and Kalman, Matan and Matias, Yossi},
  booktitle={International Conference on Machine Learning},
  pages={19274--19286},
  year={2023},
  organization={PMLR}
}

@inproceedings{Nallapati_Zhou_dos,   
title={Abstractive Text Summarization Using Sequence-to-Sequence RNNs and Beyond},  url={http://dx.doi.org/10.18653/v1/k16-1028},  DOI={10.18653/v1/k16-1028},  booktitle={Proceedings of The 20th SIGNLL Conference on Computational Natural Language Learning},  author={Nallapati, Ramesh and Zhou, Bowen and dos Santos, Cicero and Gulcehre, Caglar and Xiang, Bing},  year={2016},  month={Jan},  language={en-US}  }

@article{Cobbe_Kosaraju_Bavarian_Hilton_Nakano_Hesse_Schulman_2021,   
title={Training Verifiers to Solve Math Word Problems},  journal={Cornell University - arXiv,Cornell University - arXiv},  author={Cobbe, Karl and Kosaraju, Vineet and Bavarian, Mohammad and Hilton, Jacob and Nakano, Reiichiro and Hesse, Christopher and Schulman, John},  year={2021},  month={Oct},  language={en-US}  }

@misc{grattafiori2024llama3herdmodels,
  title={The Llama 3 Herd of Models},
  author={Aaron Grattafiori and Abhimanyu Dubey and Abhinav Jauhri and Abhinav Pandey and Abhishek Kadian and Ahmad Al-Dahle and Aiesha Letman and Akhil Mathur and Alan Schelten and Alex Vaughan and Amy Yang and Angela Fan and Anirudh Goyal et al.},
  year={2024},
  eprint={2407.21783},
  archivePrefix={arXiv},
  primaryClass={cs.AI},
  url={https://arxiv.org/abs/2407.21783}
}

@article{shen2025speculative,
  title={Speculative Decoding via Hybrid Drafting and Rollback-Aware Branch Parallelism},
  author={Shen, Yuhao and Shen, Junyi and Kong, Quan and Liu, Tianyu and Lu, Yao and Wang, Cong},
  journal={arXiv preprint arXiv:2506.01979},
  year={2025}
}

@inproceedings{rest-he-2024,
  title={REST: Retrieval-Based Speculative Decoding},
  author={He, Zhenyu and Zhong, Zexuan and Cai, Tianle and Lee, Jason and He, Di},
  booktitle={Proceedings of the 2024 Conference of the North American Chapter of the Association for Computational Linguistics: Human Language Technologies (Volume 1: Long Papers)},
  pages={1582--1595},
  year={2024}
}

@article{spec-bench-xia-2024,
  title={Unlocking efficiency in large language model inference: A comprehensive survey of speculative decoding},
  author={Xia, Heming and Yang, Zhe and Dong, Qingxiu and Wang, Peiyi and Li, Yongqi and Ge, Tao and Liu, Tianyu and Li, Wenjie and Sui, Zhifang},
  journal={arXiv preprint arXiv:2401.07851},
  year={2024}
}

@article{mt-bench-zheng-2023,
  title={Judging llm-as-a-judge with mt-bench and chatbot arena},
  author={Zheng, Lianmin and Chiang, Wei-Lin and Sheng, Ying and Zhuang, Siyuan and Wu, Zhanghao and Zhuang, Yonghao and Lin, Zi and Li, Zhuohan and Li, Dacheng and Xing, Eric and others},
  journal={Advances in Neural Information Processing Systems},
  volume={36},
  pages={46595--46623},
  year={2023}
}

@article{humaneval-chen-2021,
  title={Evaluating large language models trained on code},
  author={Chen, Mark and Tworek, Jerry and Jun, Heewoo and Yuan, Qiming and Pinto, Henrique Ponde De Oliveira and Kaplan, Jared and Edwards, Harri and Burda, Yuri and Joseph, Nicholas and Brockman, Greg and others},
  journal={arXiv preprint arXiv:2107.03374},
  year={2021}
}

@article{yang2025qwen3,
  title={Qwen3 technical report},
  author={Yang, An and Li, Anfeng and Yang, Baosong and Zhang, Beichen and Hui, Binyuan and Zheng, Bo and Yu, Bowen and Gao, Chang and Huang, Chengen and Lv, Chenxu and others},
  journal={arXiv preprint arXiv:2505.09388},
  year={2025}
}

@article{taori2023alpaca,
  title={Alpaca: A strong, replicable instruction-following model},
  author={Taori, Rohan and Gulrajani, Ishaan and Zhang, Tianyi and Dubois, Yann and Li, Xuechen and Guestrin, Carlos and Liang, Percy and Hashimoto, Tatsunori B},
  journal={Stanford Center for Research on Foundation Models. https://crfm. stanford. edu/2023/03/13/alpaca. html},
  volume={3},
  number={6},
  pages={7},
  year={2023}
}

@article{singh2025openai,
  title={OpenAI GPT-5 System Card},
  author={Singh, Aaditya and Fry, Adam and Perelman, Adam and Tart, Adam and Ganesh, Adi and El-Kishky, Ahmed and McLaughlin, Aidan and Low, Aiden and Ostrow, AJ and Ananthram, Akhila and others},
  journal={arXiv preprint arXiv:2601.03267},
  year={2025}
}

@techreport{gemini2025modelcard,
  title       = {Gemini 3 Pro Model Card},
  author      = {{Google DeepMind Team}},
  year        = {2025},
  institution = {Google DeepMind},
  url         = {https://storage.googleapis.com/deepmind-media/Model-Cards/Gemini-3-Pro-Model-Card.pdf},
  type        = {Model Card}
}

@inproceedings{li2025eagle,
  title={Eagle-3: Scaling up inference acceleration of large language models via training-time test},
  author={Li, Yuhui and Wei, Fangyun and Zhang, Chao and Zhang, Hongyang},
  booktitle={Advances in Neural Information Processing Systems},
  year={2025},
  url={https://openreview.net/forum?id=4exx1hUffq}
}

@article{liu2026talon,
  title={TALON: Confidence-Aware Speculative Decoding with Adaptive Token Trees},
  author={Liu, Tianyu and Lv, Qitan and Shen, Yuhao and Sun, Xiao and Sun, Xiaoyan},
  journal={arXiv preprint arXiv:2601.07353},
  year={2026}
}

@article{hong2025inference,
  title={Inference-Cost-Aware Dynamic Tree Construction for Efficient Inference in Large Language Models},
  author={Hong, Yinrong and Tan, Zhiquan and Hu, Kai},
  journal={arXiv preprint arXiv:2510.26577},
  year={2025}
}

@article{wu2025tetris,
  title={TETRIS: Optimal Draft Token Selection for Batch Speculative Decoding},
  author={Wu, Zhaoxuan and Zhou, Zijian and Verma, Arun and Prakash, Alok and Rus, Daniela and Low, Bryan Kian Hsiang},
  journal={arXiv preprint arXiv:2502.15197},
  year={2025}
}

@article{liu2025turbospec,
  title={TurboSpec: Closed-loop Speculation Control System for Optimizing LLM Serving Goodput},
  author={Liu, Xiaoxuan and Park, Jongseok and Hu, Langxiang and Kwon, Woosuk and Li, Zhuohan and Zhang, Chen and others},
  journal={arXiv preprint arXiv:2406.14066},
  year={2025}
}

@inproceedings{huang2025adaspec,
  title={AdaSpec: Adaptive Speculative Decoding for Fast, SLO-Aware Large Language Model Serving},
  author={Huang, Kaiyu and Wu, Hao and Shi, Zhubo and Zou, Han and Yu, Minchen and Shi, Qingjiang},
  booktitle={Proceedings of the 2025 ACM Symposium on Cloud Computing},
  pages={361--374},
  year={2025}
}

@article{sadhukhan2024magicdec,
  title={Magicdec: Breaking the latency-throughput tradeoff for long context generation with speculative decoding},
  author={Sadhukhan, Ranajoy and Chen, Jian and Chen, Zhuoming and Tiwari, Vashisth and Lai, Ruihang and Shi, Jinyuan and Yen, Ian En-Hsu and May, Avner and Chen, Tianqi and Chen, Beidi},
  journal={arXiv preprint arXiv:2408.11049},
  year={2024}
}

@inproceedings{ligumiho,
  title={Gumiho: A Hybrid Architecture to Prioritize Early Tokens in Speculative Decoding},
  author={Li, Jinze and Xu, Yixing and Huang, Haiduo and Yin, Xuanwu and Li, Dong and Ngai, Edith CH and Barsoum, Emad},
  booktitle={Forty-second International Conference on Machine Learning},
  year={2025}
}

@inproceedings{anknerhydra,
  title={Hydra: Sequentially-Dependent Draft Heads for Medusa Decoding},
  author={Ankner, Zachary and Parthasarathy, Rishab and Nrusimha, Aniruddha and Rinard, Christopher and Ragan-Kelley, Jonathan and Brandon, William},
  booktitle={First Conference on Language Modeling}
}

@article{Liu2025Illusion,
  title={Speculative Decoding: Performance or Illusion?},
  author={Liu, Xiaoxuan and Yu, Jiaxiang and Park, Jongseok and Stoica, Ion and Cheung, Alvin},
  journal={arXiv preprint arXiv:2601.11580},
  year={2025}
}

@article{brown2024dynamic,
  title={Dynamic depth decoding: Faster speculative decoding for llms},
  author={Brown, Oscar and Wang, Zhengjie and Do, Andrea and Mathew, Nikhil and Yu, Cheng},
  journal={arXiv preprint arXiv:2409.00142},
  year={2024}
}

@article{hanley1982meaning,
  title={The meaning and use of the area under a receiver operating characteristic (ROC) curve.},
  author={Hanley, James A and McNeil, Barbara J},
  journal={Radiology},
  volume={143},
  number={1},
  pages={29--36},
  year={1982}
}

@inproceedings{xiao2023smoothquant,
  title={Smoothquant: Accurate and efficient post-training quantization for large language models},
  author={Xiao, Guangxuan and Lin, Ji and Seznec, Mickael and Wu, Hao and Demouth, Julien and Han, Song},
  booktitle={International conference on machine learning},
  pages={38087--38099},
  year={2023},
  organization={PMLR}
}

@article{lin2024awq,
  title={Awq: Activation-aware weight quantization for on-device llm compression and acceleration},
  author={Lin, Ji and Tang, Jiaming and Tang, Haotian and Yang, Shang and Chen, Wei-Ming and Wang, Wei-Chen and Xiao, Guangxuan and Dang, Xingyu and Gan, Chuang and Han, Song},
  journal={Proceedings of machine learning and systems},
  volume={6},
  pages={87--100},
  year={2024}
}

@misc{vicuna2023,
    title = {Vicuna: An Open-Source Chatbot Impressing GPT-4 with 90\%* ChatGPT Quality},
    url = {https://lmsys.org/blog/2023-03-30-vicuna/},
    author = {Chiang, Wei-Lin and Li, Zhuohan and Lin, Zi and Sheng, Ying and Wu, Zhanghao and Zhang, Hao and Zheng, Lianmin and Zhuang, Siyuan and Zhuang, Yonghao and Gonzalez, Joseph E. and Stoica, Ion and Xing, Eric P.},
    month = {March},
    year = {2023}
}

@article{wang2025opt,
  title={Opt-tree: Speculative decoding with adaptive draft tree structure},
  author={Wang, Jikai and Su, Yi and Li, Juntao and Xia, Qingrong and Ye, Zi and Duan, Xinyu and Wang, Zhefeng and Zhang, Min},
  journal={Transactions of the Association for Computational Linguistics},
  volume={13},
  pages={188--199},
  year={2025},
  publisher={MIT Press 255 Main Street, 9th Floor, Cambridge, Massachusetts 02142, USA~…}
}

@article{li2025adaserve,
  title={Adaserve: Slo-customized llm serving with fine-grained speculative decoding},
  author={Li, Zikun and Chen, Zhuofu and Delacourt, Remi and Oliaro, Gabriele and Wang, Zeyu and Chen, Qinghan and Lin, Shuhuai and Yang, April and Zhang, Zhihao and Chen, Zhuoming and others},
  journal={arXiv e-prints},
  pages={arXiv--2501},
  year={2025}
}

@article{zhou2023distillspec,
  title={Distillspec: Improving speculative decoding via knowledge distillation},
  author={Zhou, Yongchao and Lyu, Kaifeng and Rawat, Ankit Singh and Menon, Aditya Krishna and Rostamizadeh, Afshin and Kumar, Sanjiv and Kagy, Jean-Fran{\c{c}}ois and Agarwal, Rishabh},
  journal={arXiv preprint arXiv:2310.08461},
  year={2023}
}

@article{zeng2025glm,
  title={Glm-4.5: Agentic, reasoning, and coding (arc) foundation models},
  author={Zeng, Aohan and Lv, Xin and Zheng, Qinkai and Hou, Zhenyu and Chen, Bin and Xie, Chengxing and Wang, Cunxiang and Yin, Da and Zeng, Hao and Zhang, Jiajie and others},
  journal={arXiv preprint arXiv:2508.06471},
  year={2025}
}

@misc{saxena2023prompt,
    title = {Prompt Lookup Decoding},
    author = {Apoorv Saxena},
    year = {2023},
    month = {November},
    url = {https://github.com/apoorvumang/prompt-lookup-decoding/}
}

@article{liu2024pearl,
  title={Pearl: Parallel speculative decoding with adaptive draft length},
  author={Liu, Tianyu and Li, Yun and Lv, Qitan and Liu, Kai and Zhu, Jianchen and Hu, Winston and Sun, Xiao},
  journal={arXiv preprint arXiv:2408.11850},
  year={2024}
}

@article{liu2023online,
  title={Online speculative decoding},
  author={Liu, Xiaoxuan and Hu, Lanxiang and Bailis, Peter and Cheung, Alvin and Deng, Zhijie and Stoica, Ion and Zhang, Hao},
  journal={arXiv preprint arXiv:2310.07177},
  year={2023}
}

@article{shen2026double,
  title={Double: Breaking the Acceleration Limit via Double Retrieval Speculative Parallelism},
  author={Shen, Yuhao and Liu, Tianyu and Shen, Junyi and Wu, Jinyang and Kong, Quan and Huan, Li and Wang, Cong},
  journal={arXiv preprint arXiv:2601.05524},
  year={2026}
}
\bibliographystyle{iclr2026_conference}

\appendix
\clearpage

\section{Theoretical Analysis}
\label{app:theory}

This appendix provides full derivations for the two guarantees summarized in the main text (Theorem~\ref{thm:width_coverage} and Theorem~\ref{thm:exchange}). We keep the exposition aligned with the serving-centric formulation in Sec.~\ref{sec:problem} and the fixed verification cap in Eq.~\ref{eq:global_budget}.

\subsection{Coverage Gain via Width Expansion}
\label{app:theory_width}

\textit{ECHO} triggers ``Opportunistic Width Expansion'' only when depth extension is halted due to low confidence and a residual budget exists. Here we formalize the probabilistic advantage of this structural shift.

As noted in Gumiho~\citep{ligumiho}, the acceptance probability of a single path decays significantly as depth increases. At a truncation depth $d$ where the top-1 confidence is low, allocating the residual budget to width does not guarantee a longer accepted sequence immediately, but it \textbf{strictly increases the cumulative probability mass} that the ground-truth token is covered within the candidate set.

\setcounter{theorem}{0}
\begin{theorem}[Coverage Gain]
\label{thm:width_coverage_app}
Let $\mathcal{S}_k = \{x^{(1)}, \dots, x^{(k)}\}$ be the set of top-$k$ candidate tokens at depth $d$, sorted by probability $p_t(x \mid x_{<d})$.
Let $x^*$ be the ground-truth token sampled from $p_t$, and define the coverage probability as $\mathbb{P}(x^* \in \mathcal{S}) \triangleq \sum_{x \in \mathcal{S}} p_t(x)$.
Expanding the candidate set size from $k$ to $k'$ (where $k' > k$) yields a strictly positive gain in the coverage probability:
\begin{equation}
\mathbb{P}(x^* \in \mathcal{S}_{k'}) - \mathbb{P}(x^* \in \mathcal{S}_k)
= \sum_{i=k+1}^{k'} p_t(x^{(i)} \mid x_{<d}) > 0,
\label{eq:coverage_gain_app}
\end{equation}
whenever $\sum_{i=k+1}^{k'} p_t(x^{(i)} \mid x_{<d})>0$.
\end{theorem}

\begin{proof}[Proof Sketch]
Since the target distribution has non-zero entropy (implied by the low confidence that triggered truncation), the probability mass is distributed across multiple tokens (i.e., $p_t(x^{(i)} \mid x_{<d}) > 0$ for some $i > k$).
Equation~\ref{eq:coverage_gain_app} explicitly shows that widening the set accumulates these non-zero probabilities.
This probabilistic gain minimizes the risk of the target model rejecting the entire node, thereby acting as a ``safety net'' when depth extension is too risky.
\end{proof}

\noindent\textbf{Remark.}
Theorem~\ref{thm:width_coverage_app} formalizes why, once depth extension becomes unreliable, spending residual budget on width is a principled fallback: it improves \emph{coverage} (probability of including $x^*$) even if it does not guarantee immediate depth progress.

\subsection{Batch-Level Objective Under Compute-Bound Constraints}
\label{app:theory_objective}

This section supports the statement used in Theorem~\ref{thm:exchange} (main text): under saturated compute-bound serving with a fixed verification cap, improving end-to-end throughput reduces to improving the batch-level expected accepted tokens per iteration.

\paragraph{Setup.}
Consider a single SD iteration with a batch of $B$ requests. Request $i$ submits $K_i$ draft candidates, and the target model accepts $L_i$ tokens (random variable). The total verification load is
$
K_{\text{total}} = \sum_{i=1}^B K_i.
$
In compute-bound serving (Sec.~\ref{sec:problem}), verification latency grows approximately linearly with the verified token count (Eq.~\ref{eq:ver_linear}). \textit{ECHO} enforces a per-iteration cap (Eq.~\ref{eq:global_budget}) such that $K_{\text{total}} \le K_{\max}$, and the system is typically operated near saturation, i.e., $K_{\text{total}} \approx K_{\max}$.

\begin{proposition}[Compute-Bound Objective Reduction]
\label{prop:cb_reduction_app}
Under a fixed verification cap $K_{\text{total}} = K_{\max}$, the per-iteration latency is (approximately) constant. Consequently, maximizing end-to-end system throughput is equivalent to maximizing the \textbf{aggregate expected accepted length} per iteration:
\begin{equation}
\max \;\mathcal{J}
\;\triangleq\;
\max \sum_{i=1}^B \mathbb{E}[L_i].
\label{eq:goodput_objective_app}
\end{equation}
\end{proposition}

\begin{proof}
Throughput (accepted tokens per unit time) for an iteration is proportional to
$
\frac{\sum_{i=1}^B L_i}{T_{\text{iter}}}.
$
In the compute-bound regime, Eq.~\ref{eq:ver_linear} implies $T_{\text{iter}}$ is dominated by verification cost and scales with $K_{\text{total}}$. If the scheduler enforces a fixed total verified token count $K_{\text{total}}=K_{\max}$ (Eq.~\ref{eq:global_budget}), then $T_{\text{iter}}$ becomes (approximately) a constant for that iteration.
Therefore, maximizing expected throughput reduces to maximizing the expected numerator, i.e., $\sum_i \mathbb{E}[L_i]$.
\end{proof}

\paragraph{Why this matters.}
Proposition~\ref{prop:cb_reduction_app} implies that a policy can improve system performance even if it reduces the MAT of some requests, provided the batch-level aggregate $\sum_i \mathbb{E}[L_i]$ increases. This motivates elastic reallocation across requests in \textit{ECHO}.

\subsection{Marginal Utility Exchange and Budget Reallocation}
\label{app:theory_exchange}

We now formalize when reallocating a small amount of budget from one request to another improves the batch objective in Eq.~\ref{eq:goodput_objective_app}.

\paragraph{Per-request response curve.}
Let
\begin{equation}
f_i(k) \triangleq \mathbb{E}[L_i \mid K_i = k]
\label{eq:fi_def_app}
\end{equation}
be the expected accepted length for request $i$ given $k$ verified candidates. Define the marginal gain of the $k$-th candidate token as
\begin{equation}
\Delta_i(k) \triangleq f_i(k) - f_i(k-1).
\label{eq:delta_def_app}
\end{equation}
Intuitively, $\Delta_i(k)$ measures the incremental expected accepted tokens obtained by spending one additional verification token on request $i$.

\begin{theorem}[Marginal Utility Exchange]
\label{thm:exchange_app}
Consider two requests $i$ and $j$ sharing a fixed total budget $K$ within an iteration, with $\sum_{m=1}^B K_m = K_{\max}$.
If the marginal gain of request $j$ at its next token exceeds the marginal gain of request $i$ at its current last token:
\begin{equation}
\Delta_j(K_j + 1) \;>\; \Delta_i(K_i),
\label{eq:exchange_cond_app}
\end{equation}
then reallocating one token from $i$ to $j$ (i.e., $K_i \leftarrow K_i-1,\, K_j \leftarrow K_j+1$) strictly increases the batch objective
$
\mathcal{J}=\sum_{m=1}^B \mathbb{E}[L_m]
$
and therefore improves end-to-end throughput under compute-bound serving (by Proposition~\ref{prop:cb_reduction_app}).
\end{theorem}

\begin{proof}
Let the original allocation be $(K_i, K_j)$ and the new allocation be $(K_i-1, K_j+1)$, with other $K_m$ unchanged.
The change in the batch objective is
\begin{align}
\Delta \mathcal{J}
&=
\bigl[f_j(K_j+1) - f_j(K_j)\bigr]
-
\bigl[f_i(K_i) - f_i(K_i-1)\bigr] \nonumber\\
&=
\Delta_j(K_j+1) - \Delta_i(K_i).
\label{eq:dJ_app}
\end{align}
Under condition~\eqref{eq:exchange_cond_app}, $\Delta \mathcal{J} > 0$, hence the batch objective strictly increases.
By Proposition~\ref{prop:cb_reduction_app}, this strictly improves throughput in the compute-bound regime with a fixed verification cap.
\end{proof}

\paragraph{Implication for \textit{ECHO}.}
Sparse gating in \textit{ECHO} is designed to identify low-confidence (hence low marginal-utility) continuations.
When confidence is low, the marginal gain $\Delta_i(\cdot)$ diminishes due to probability decay along deeper paths; truncating such branches reduces spending on low-utility tokens.
The freed budget can be (i) reallocated to other high-confidence requests for further depth extension, or (ii) used for width expansion at a truncation depth when no depth extension is available.
Both behaviors align with Theorem~\ref{thm:exchange_app}: moving budget from lower $\Delta$ to higher $\Delta$ increases the batch-level objective, even if the MAT of some truncated requests becomes shorter.

\newpage

\section{Extended Related Works}
\label{sec:related_work}

\subsection{Speculative Decoding}
SD exploits the parallel verification capability of Transformers to accelerate autoregressive generation~\citep{leviathan2023fast, chen2023accelerating,shen2025speculative}. The core paradigm involves a lightweight drafter proposing a sequence of tokens, which the target model verifies in a single forward pass.
Draft model-based methods utilize a smaller, separate model (often quantized or distilled) to generate candidates~\citep{zhou2023distillspec}. While effective, they require maintaining two separate models and synchronizing their vocabularies.
To mitigate this maintenance overhead, draft model-free methods have emerged. Approaches like Medusa~\citep{cai2024medusa} and EAGLE~\citep{li2024eagle, li2025eagle} attach auxiliary prediction heads to the target model's frozen layers to predict future tokens. Recently, Multi-Token Prediction (MTP)~\citep{zeng2025glm} co-trains these heads jointly with the main model.
Additionally, non-parametric methods leverage retrieval or matching logic without additional training. Notable examples include prompt lookup decoding~\citep{saxena2023prompt}, which reuses recurring phrases from the context, and retrieval-augmented SD~\citep{rest-he-2024,shen2026double}, which fetches candidates from external corpora.
While the aforementioned methods focus on the generation source of draft tokens, maximizing the verification efficiency further requires optimizing the structural organization of these candidates.

\subsection{Dynamic Token Tree}
To further improve the acceptance rate per verification step, tree-based speculation explores multiple candidate paths simultaneously~\citep{miao2024specinfer}.
Static tree methods rely on predetermined structures, where the tree shape (width and depth) is fixed manually or heuristically based on average acceptance rates~\citep{cai2024medusa, li2025eagle}. These methods prioritize high-probability tokens to form rigid layers but lack flexibility.
Dynamic tree methods aim to adapt the tree structure to the current generation context. Approaches like TALON~\citep{liu2026talon} and OPT-Tree~\citep{wang2025opt} construct larger trees or dynamic widths based on token entropy and probability, applying adaptive pruning to fit within a compute budget.
However, these methods typically employ dense, node-wise evaluation logic that incurs significant control overhead and often generates irregular (ragged) batch shapes incompatible with standard high-performance serving kernels.

\subsection{SD in High-Concurrency Scenarios}
With the rise of production LLM serving systems, research has expanded from single-batch acceleration to system-level scheduling.
Existing scheduling works can be broadly categorized into server-side approaches~\citep{fu2024break, liu2023online}, which focus on maximizing overall throughput and device utilization, and client-side approaches~\citep{liu2024pearl}, which optimize for user-perceived latency and fairness.
Recent works like TETRIS~\citep{wu2025tetris} and TurboSpec~\citep{liu2025turbospec} have begun to integrate SD into serving systems, optimizing draft token selection to balance inference speed with profitability.
However, most prior work treats the draft tree as a black box or ignores the compute-bound nature of verification in high-concurrency regimes. \textit{ECHO} bridges this gap by proposing a principled, serving-centric framework that jointly optimizes tree construction and budget scheduling under strict verification constraints.

\newpage

\section{Evaluation Details} \label{sec:appendix_eval}
For reproducibility, we discuss the experimental setup (Section~\ref{sec:experiments}) in detail and the source code of this project will be made available at a later time.

\subsection{Data Configurations}

In our experiments, we evaluate \textit{ECHO} using the following dataset settings. To ensure comprehensive coverage, our benchmarks span code generation, mathematical reasoning, summarization, and general instruction following: HumanEval~\citep{humaneval-chen-2021}, GSM8K~\citep{Cobbe_Kosaraju_Bavarian_Hilton_Nakano_Hesse_Schulman_2021}, CNN/DM~\citep{Nallapati_Zhou_dos}, Alpaca~\citep{taori2023alpaca}, and MT-Bench~\citep{mt-bench-zheng-2023} following the set of EAGLE-3. The maximum generation length for these tasks is set to 1024 tokens.

\subsection{Detailed Baselines}
\paragraph{Baselines and Implementation}
To strictly evaluate the effectiveness of \textit{ECHO} in production-grade serving, we benchmark it against a comprehensive set of competitive baselines, categorizing them into four distinct paradigms of speculative decoding:

\begin{itemize}[leftmargin=*, itemsep=0pt, topsep=2pt]
    \item \textbf{Standard SD}~\citep{chen2023accelerating,leviathan2023fast}: The fundamental \textit{draft-then-verify} framework. We utilize the same draft models as \textit{ECHO} but execute them sequentially without any tree-based parallel verification, serving as the baseline to measure the raw speedup gain over auto-regressive decoding.
    
    \item \textbf{Retrieval-based SD}: We compare against \textbf{Lookahead}~\citep{fu2024break}, a representative method that accelerates inference solely using the target model. Lookahead employs Jacobi iteration to generate multi-branch candidates without requiring a separate draft model, providing a reference for architecture-agnostic acceleration.
    
    \item \textbf{Training-based Methods}: This category represents the current state-of-the-art. We include MLP-based approaches like \textbf{Medusa}~\citep{cai2024medusa} and \textbf{Hydra}~\citep{anknerhydra}, which attach lightweight decoding heads to predict multiple tokens in parallel. Crucially, we select \textbf{EAGLE-3}~\citep{li2025eagle} as our primary static baseline. EAGLE-3 utilizes feature-level auto-regression with a multi-layer fusion mechanism and constructs a static draft tree with fixed depth and width. Comparing against EAGLE-3 allows us to directly demonstrate the advantages of \textit{ECHO}'s elastic budget scheduling over rigid geometric constraints in high-concurrency regimes.
    
    \item \textbf{Dynamic Tree Methods}: To assess the efficacy of our sparse gating strategy, we compare against \textbf{OPT-Tree}~\citep{wang2025opt} and \textbf{DDD}~\citep{brown2024dynamic}. These methods optimize the draft tree topology typically via dense, node-wise heuristics or "generate-then-prune" strategies. \textbf{Note on Fairness:} Since these methods were originally designed for EAGLE-2, we have re-implemented and adapted them to support the stronger EAGLE-3 backbone. This ensures that any performance gap is attributable to the tree scheduling policy (Dense Control vs. Sparse Gating) rather than the underlying draft model capability.
\end{itemize}

\subsection{Model Configurations}
To validate performance, we select state-of-the-art open-source model pairs followed by EAGLE3 such as the Vicuna (yuhuili/EAGLE3-Vicuna1.3-13B, lmsys/vicuna-13b-v1.3), LLaMA3.1 (yuhuili/EAGLE3-LLaMA3.1-Instruct-8B, meta-llama/Llama-3.1-8B-Instruct), LLaMA3.3 (yuhuili/EAGLE3-LLaMA3.3-Instruct-70B, meta-llama/Llama-3.3-70B-Instruct), Qwen3-8B (AngelSlim/Qwen3-8B-eagle3, Qwen/Qwen3-8B), Qwen3-32B (AngelSlim/Qwen3-32B-eagle3, Qwen/Qwen3-32B) and Qwen3-235B (lmsys/Qwen3-235B-A22B-EAGLE3, Qwen/Qwen3-235B-A22B) for each task. All model weights are loaded in bfloat16 format for optimized GPU inference without quantization. As a training-free method, \textit{ECHO} does not modify any draft model parameters during evaluation. We summarize the model configuration in Table~\ref{tab:model_conf}. 

\begin{table}[h]
\caption{Model configurations.}
\centering
\resizebox{0.6\columnwidth}{!}{  
\begin{tabular}{ccccc}
\toprule
    \textbf{Models} & \textbf{Layers} & \textbf{dim} & \textbf{FFN dim} & \textbf{Vocabulary size} \\
\midrule
Vicuna 68M & 2  & 768  & 3072  & 32000    \\
Vicuna 13B & 40  & 5120  & 13824  & 32000  \\
LLaMA-3.1 8B & 32 & 4096 & 14336  & 128256   \\
LLaMA-3.3 70B & 80 & 8192  & 28672  & 128256 \\
Qwen-3 8B & 36 & 4096  & 12288  & 151936 \\
Qwen-3 32B & 64 & 5120  & 25600 & 151936 \\
Qwen-3 235B & 94 & 4096  & 12288 & 151936 \\
\bottomrule
\end{tabular}
}
\label{tab:model_conf}
\end{table}

\subsection{Evaluation Details}
\textbf{Hardware and Framework.}
All experiments are conducted on 8 NVIDIA H100 (80GB) GPUs. For latency-critical benchmarks ($BS=1$), we utilize the HuggingFace \texttt{transformers} library. For high-concurrency scenarios evaluations ($BS>1$), we integrate all methods into \texttt{SGLang}, an industrial-grade inference engine, to accurately measure throughput under realistic kernel constraints. We use greedy sampling ($\text{temperature}=0$) for all methods to ensure deterministic reproducibility. Specifically, we test qwen235b in Figure.\ref{fig:batch_scaling} without \textit{CUDA-Graph} and other experiments in high batch size settings are all using the \textit{CUDA-Graph}.

\textbf{Adaptive Calibration.}
To ensure robustness and generalization across diverse model-dataset combinations, \textit{ECHO} incorporates a lightweight warm-up phase prior to evaluation. During this phase, we analyze the layer-wise acceptance distribution of the draft model to adaptively identify discriminative "sweet spots" and calibrate gating thresholds. This mechanism allows \textit{ECHO} to automatically tailor its sparse gating policy to the specific confidence landscape of each model. For instance, the calibrated thresholds $\tau_d$ at specific depths $d$ are set as follows: for LLaMA-3.1-8B, $\{d_0: 0.2, d_5: 0.35, d_8: 0.5\}$; and for the larger Qwen3-235B, $\{d_0: 0.15, d_3: 0.3, d_5: 0.5\}$.

\textbf{Configuration Protocols.}
We adopt distinct configuration strategies for different concurrency regimes:
\begin{itemize}[leftmargin=*, itemsep=0pt, topsep=2pt]
    \item \textbf{Low Concurrency ($BS=1$):} Both EAGLE-3 and \textit{ECHO} employ the default configuration (Tree Depth=8, Top-$k$=10, Total Tokens=60) to maximize the theoretical upper bound of acceptance length.
    \item \textbf{High Concurrency ($BS>1$):} EAGLE3 is configured with their respective default parameters. In contrast, \textit{ECHO} initializes with a streamlined configuration (Tree Depth=3, Top-$k$=3, Total Tokens=5) and activates its elastic scheduler. This allows \textit{ECHO} to dynamically adjust the depth and width allocation per request based on the real-time batch size and global verification budget, ensuring optimal resource utilization in compute-bound regimes.
\end{itemize}

\newpage

\section{Layer-wise Confidence Density Visualization for Each Model}

\subsection{Visualization of Layer-wise Confidence Shifts}
\label{sec:appendix_visualization}

We provide a comprehensive visualization of layer-wise confidence distributions to characterize the stochastic nature of drafting across varying architectures. As illustrated in Figures~\ref{fig:layer6} through \ref{fig:layer9} , we profile the probability density of \textcolor{red}{Accepted (Red)} and \textcolor{blue}{Rejected (Blue)} tokens up to \textbf{Depth=8} for the LLaMA series and \textbf{Depth=5} for the Qwen series.

\paragraph{The Stochastic Drift.}
These visualizations serve as the empirical foundation for \textit{ECHO}. A critical oversight in prior dynamic methods is the assumption of a static confidence threshold. The plots reveal a distinct \textbf{Entropy Drift}: at shallow depths (e.g., Depth 0-1), the distributions are bimodal with a clear decision boundary, accepted tokens cluster near probability $0.9$, while rejected tokens cluster near $0.1$. However, as the tree deepens, the confidence of valid tokens naturally decays, causing the red distribution to migrate leftward and merge with the blue distribution. This creates the \textbf{"Overlap Regions"} (visualized in purple/grey), high-entropy zones where binary classification becomes inherently ambiguous. Operating blindly in these regions leads to error accumulation, manifesting as either the premature pruning of valid paths (false negatives) or the wasteful extension of invalid ones (false positives).

\paragraph{Analysis of Distributional Dynamics.}
The density landscapes reveal distinct behaviors governed by model architecture and alignment:
\begin{itemize}[leftmargin=*]
    \item \textbf{Depth-Dependent Decay:} A universal trend is observed where the probability mass of the optimal path migrates from the high-confidence "Accepted" cluster toward the "Rejected" cluster as depth increases, blurring the separability.
    \item \textbf{Alignment Sensitivity:} The velocity of this drift correlates with the intrinsic draft-target alignment. Models with weaker alignment exhibit a more rapid transition into the ambiguous overlap zone at shallower depths.
\end{itemize}

\paragraph{Operationalizing the "Sweet Spot".}
This variance substantiates the design of our \textbf{Sparse Gating}. As explicitly annotated by the \textcolor{green}{\textbf{Green Dashed Lines}} in the figures, \textit{ECHO} identifies \textbf{"Sweet Spots"} layers where the overlap is minimal and the signal-to-noise ratio is maximized. Instead of applying a uniform threshold, \textit{ECHO} utilizes a warm-up phase to profile these specific distributional shifts. By adaptively triggering verification gates only at these identified "Sweet Spots" and dynamically calibrating thresholds ($\tau_d$) to the local distribution, \textit{ECHO} circumvents the high-entropy overlap zones, ensuring robust generalization across diverse model-dataset combinations.
\newpage
\begin{figure*}[h]
\begin{center}
\centerline{\includegraphics[width=0.9\columnwidth]{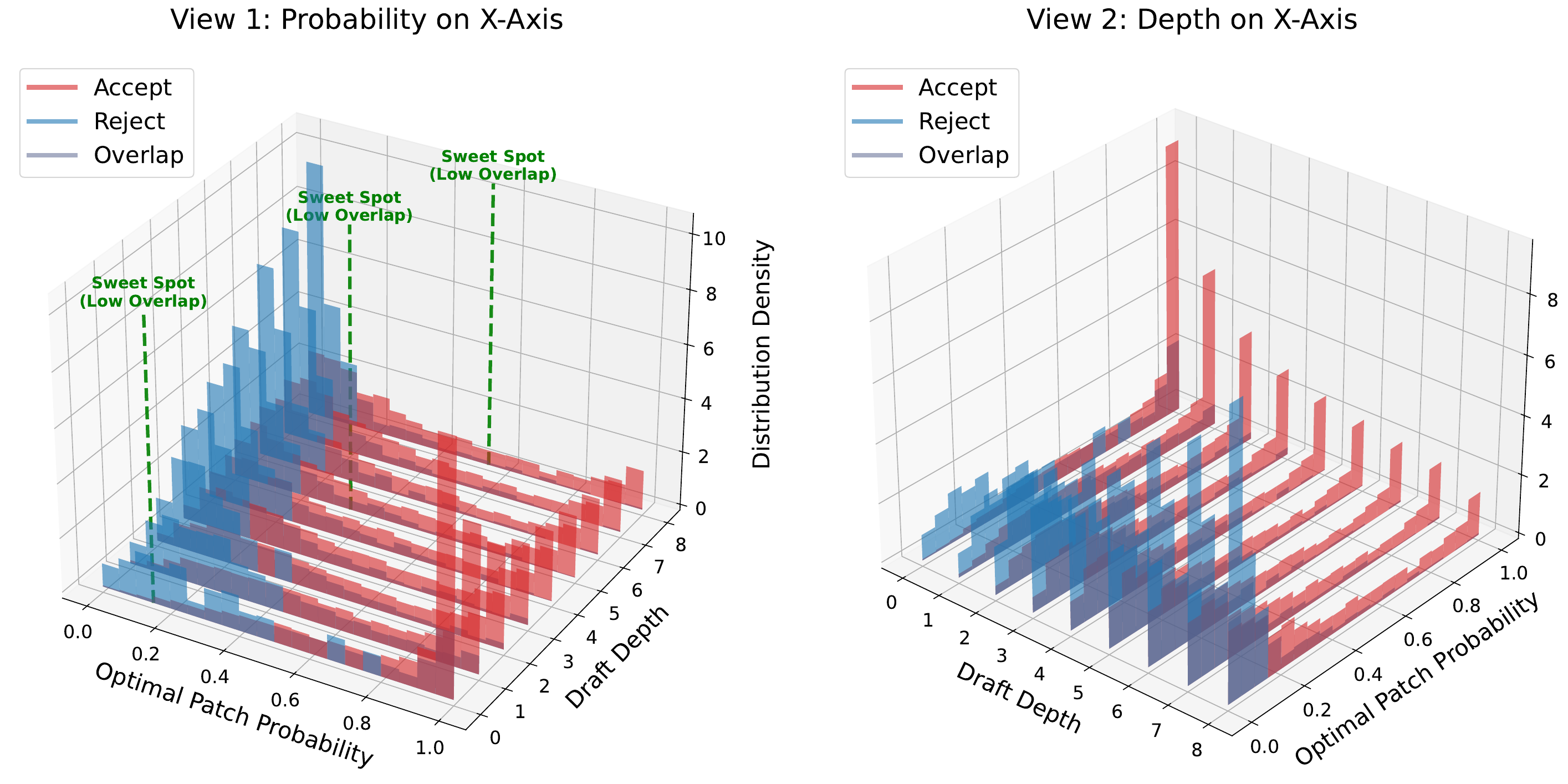}} 
\caption{Visualization of confidence distributions across
draft depths (LLaMA3.1-8B on MT-Bench).}
\label{fig:layer6}
\end{center}
\vspace*{-0.4in}
\end{figure*}

\begin{figure*}[h]
\begin{center}
\centerline{\includegraphics[width=0.9\columnwidth]{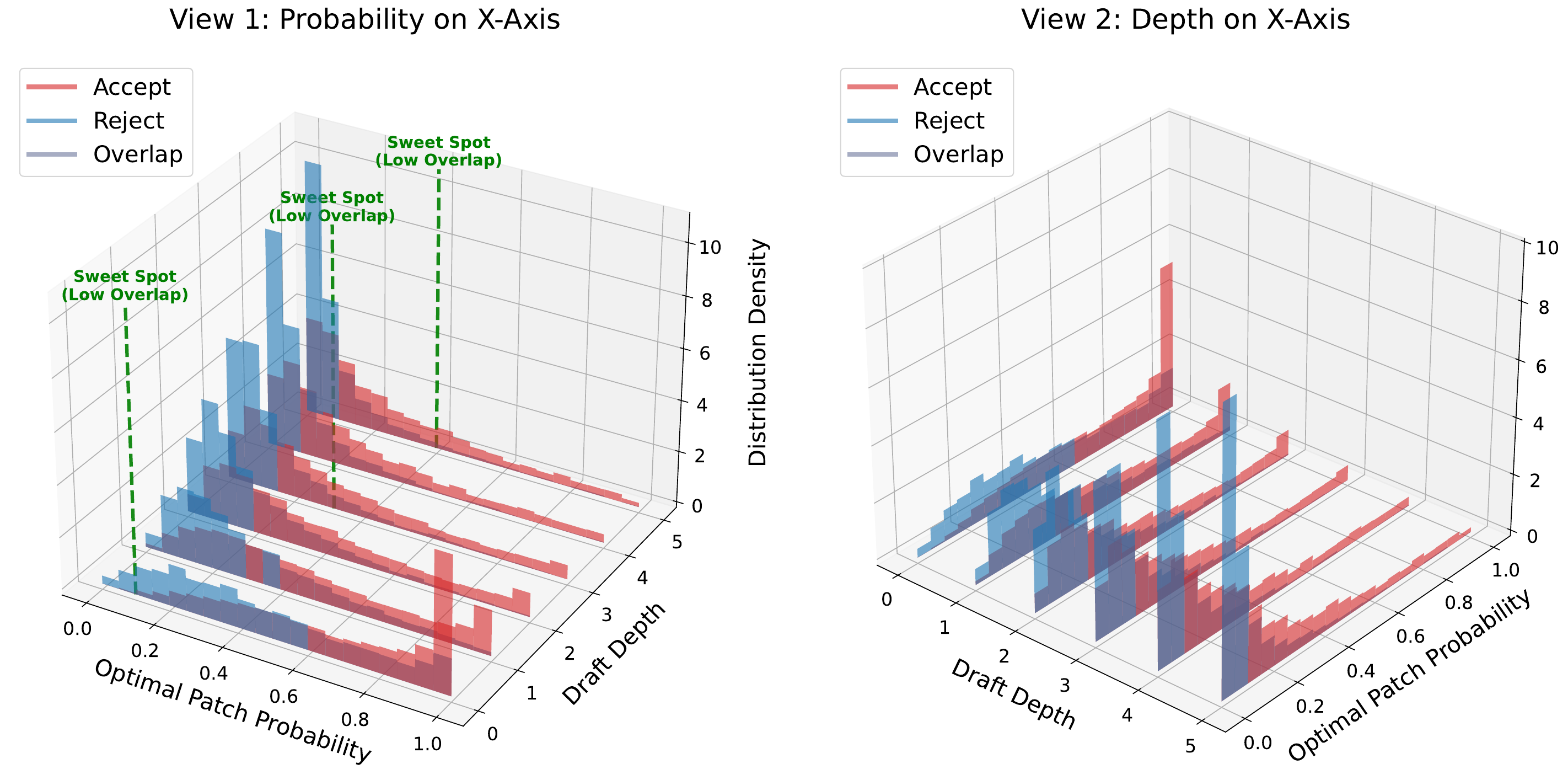}} 
\caption{Visualization of confidence distributions across
draft depths (Qwen3-8B on MT-Bench).}
\label{fig:layer7}
\end{center}
\vspace*{-0.4in}
\end{figure*}

\begin{figure*}[h]
\begin{center}
\centerline{\includegraphics[width=0.9\columnwidth]{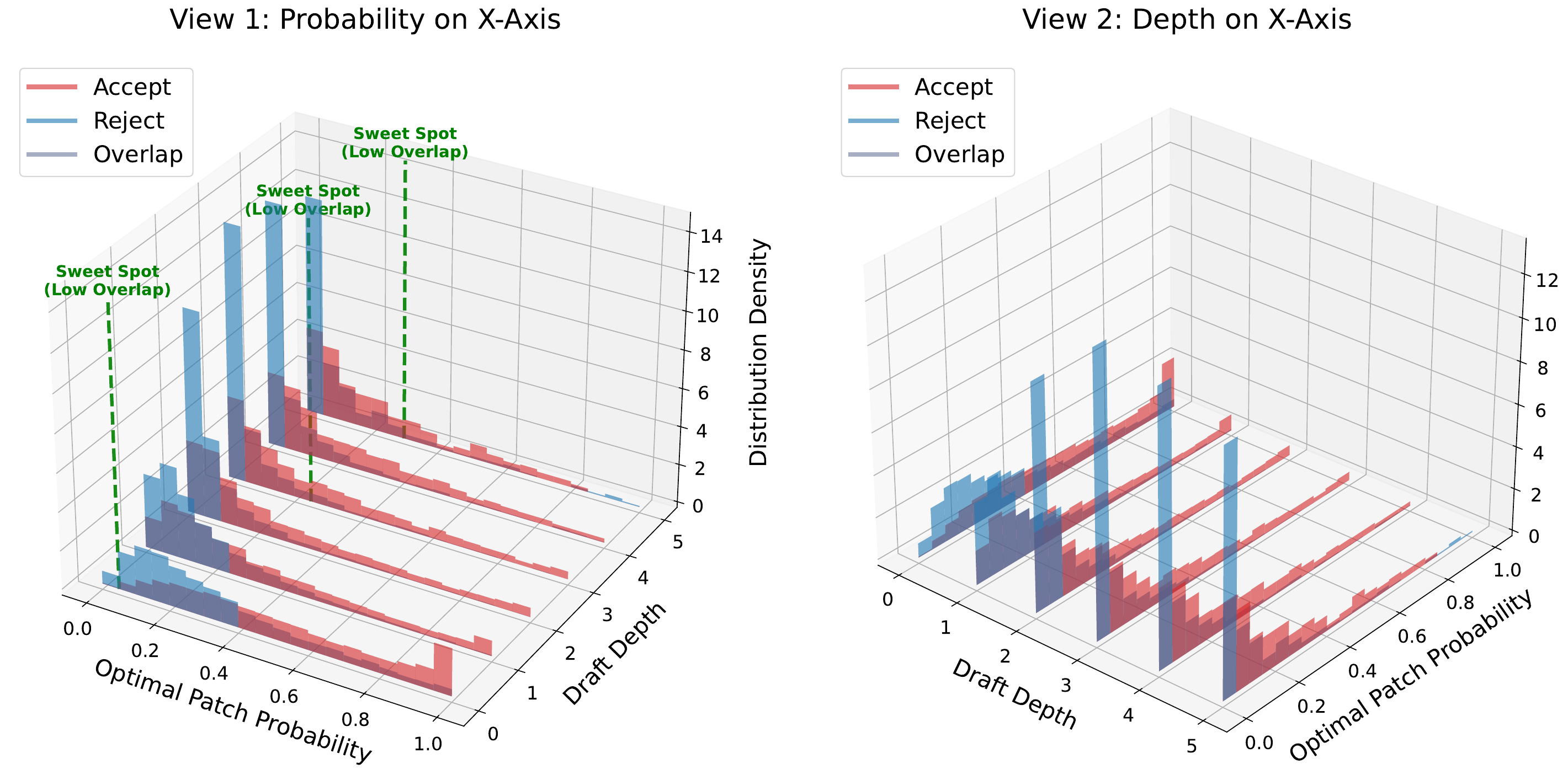}} 
\caption{Visualization of confidence distributions across
draft depths (Qwen3-32B on MT-Bench).}
\label{fig:layer8}
\end{center}
\vspace*{-0.2in}
\end{figure*}

\begin{figure*}[t]
\begin{center}
\centerline{\includegraphics[width=0.9\columnwidth]{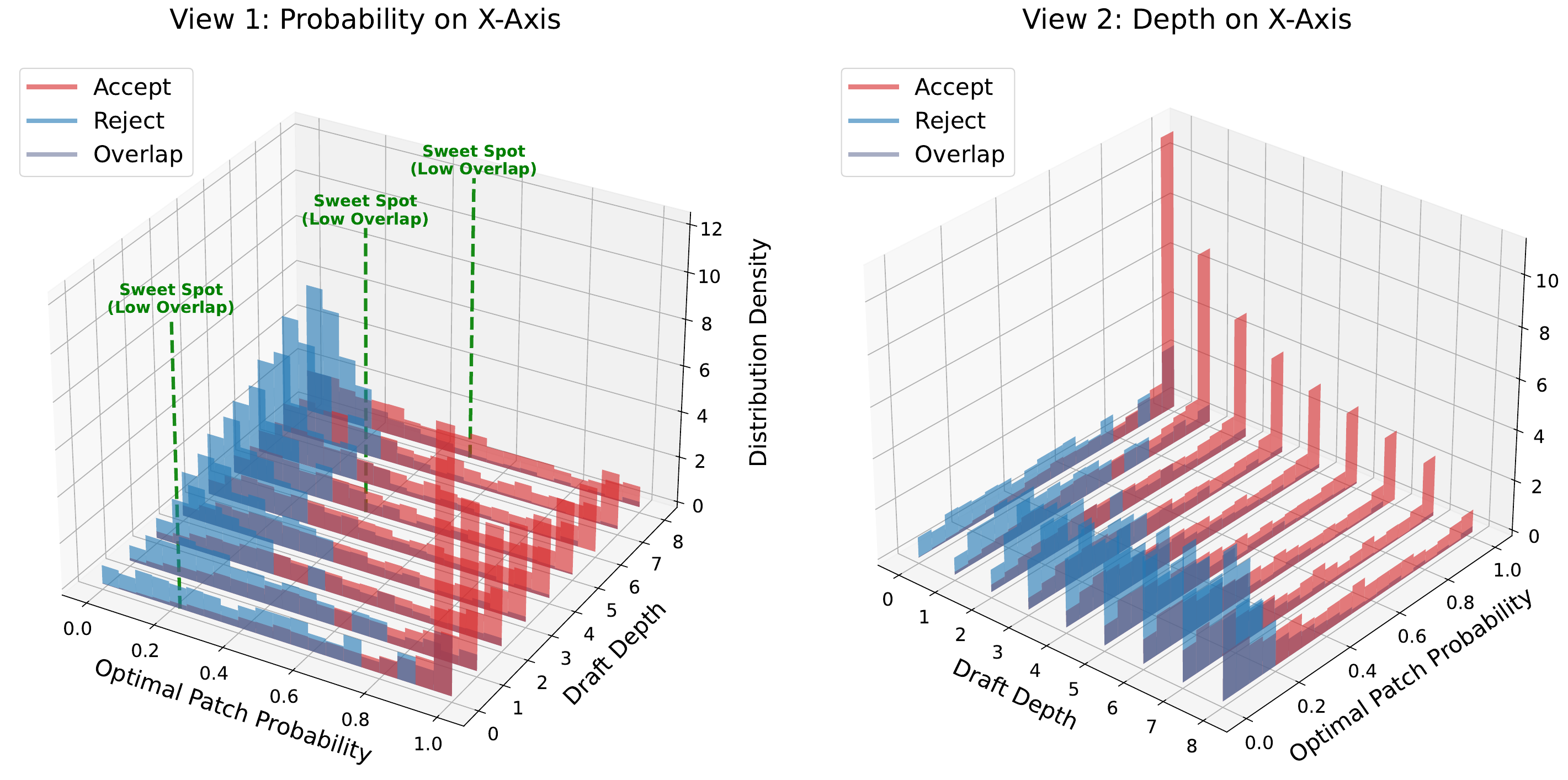}} 
\caption{Visualization of confidence distributions across
draft depths (LLaMA3.3-70B on MT-Bench).}
\label{fig:layer9}
\end{center}
\vspace*{-0.4in}
\end{figure*}


\end{document}